              \def\version{\today}	        	%

\documentclass[reqno,11pt]{amsart} 
\usepackage[T1]{fontenc}
\usepackage[utf8]{inputenc}
\usepackage{amsmath,amsthm} 
\usepackage[normalem]{ulem}

\usepackage[mathscr]{eucal}
\usepackage{amssymb}
\usepackage{srcltx} 
\usepackage{dsfont}
\usepackage{hyperref}
\usepackage{xcolor}
\usepackage{enumerate}
\usepackage{tikz, pgfplots}
\usetikzlibrary{patterns}
\usepackage{caption}
\usepackage{subcaption}
\usepackage{comment}
\usepackage{lscape}
\usepackage{graphicx}
\usepackage{tabularx}
\numberwithin{equation}{section}

\def\emptyset{\varnothing} 

\newfam\Bbbfam 
\font\tenBbb=msbm10 
\font\sevenBbb=msbm7 
\font\fiveBbb=msbm5 
\textfont\Bbbfam=\tenBbb 
\scriptfont\Bbbfam=\sevenBbb 
\scriptscriptfont\Bbbfam=\fiveBbb

\def\2{\mathbf 2}

\def\barn1a{\bar n_{1a}}
\def\tilden3{\widetilde n_3}

\newcommand{\R}     {\mathbb{R}} 
 
\newcommand{\N}     {\mathbb{N}} 
\renewcommand{\P}   {\mathbb{P}}

\def\1{{\mathchoice {1\mskip-4mu\mathrm l}      
{1\mskip-4mu\mathrm l} 
{1\mskip-4.5mu\mathrm l} {1\mskip-5mu\mathrm l}}} 
 
\def\comment#1{} 
\newtheoremstyle{thm}{2ex}{2ex}{\itshape\rmfamily}{} 
{\bfseries\rmfamily}{}{1.7ex}{} 
 
\newtheoremstyle{rem}{1.3ex}{1.3ex}{\rmfamily}{} 
{\itshape\rmfamily}{}{1.5ex}{}

\renewcommand{\theequation}{\thesection.\arabic{equation}} 
 
\newtheorem{theorem}{Theorem}[section] 
\newtheorem{lemma}[theorem]{Lemma}

\newtheorem{conj}[theorem] {Conjecture}

\theoremstyle{definition}
\newtheorem{defn}[theorem] {Definition}

\newtheorem{remark}[theorem]{Remark}

\newcommand{\eps}{\varepsilon}

\newcommand\numberthis{\addtocounter{equation}{1}\tag{\theequation}}

\definecolor{Red}{rgb}{1,0,0}

\pgfplotsset{compat=1.18}

\definecolor{gb}{rgb}{0.0, 0.36, 0.56}
\newcommand{\gb}{\color{gb}}  
\newcommand{\igb}{\iffalse\gb\fi}
\newcommand{\bl}{\color{black}}
\definecolor{lightgreen}{rgb}{0.8, 0.9, 0.1}
\definecolor{darkgreen}{rgb}{0.6, 0.8, 0.3}
\definecolor{red1}{rgb}{0.9, 0, 0.1}
\definecolor{purple1}{rgb}{0.5, 0, 0.6}
\definecolor{blue1}{rgb}{0.4, 0.6, 0.99}
\definecolor{orange1}{rgb}{0.9, 0.7, 0}

\begin{document} 

\title[The interplay of selection and dormancy in a Moran model]{The interplay of selection and dormancy in a Moran model can lead to \igb  coexistence of types}


\author[Jochen Blath, Baptiste Le Duigou and András Tóbiás]{}
\maketitle
\centerline{\sc Jochen Blath{\footnote{Goethe-Universität Frankfurt, Institute of Mathematics, and $\mathrm{C^3S}$ -- Center for Critical Computational Studies, Robert-Mayer-Straße 10, 60325 Frankfurt am Main, Germany, {\tt blath@math.uni-frankfurt.de}}}, Baptiste Le Duigou{\footnote{Department of Computer Science and Information Theory, Faculty of Electrical Engineering and Informatics, Budapest University of Technology and Economics, Műegyetem rkp. 3., H-1111 Budapest, Hungary, and HUN-REN Alfréd Rényi Institute of Mathematics, Reáltanoda utca 13--15, H-1053 Budapest, Hungary,  {\tt leduigou@cs.bme.hu}}} and András Tóbiás{\footnote{Department of Computer Science and Information Theory, Faculty of Electrical Engineering and Informatics, Budapest University of Technology and Economics, Műegyetem rkp. 3., H-1111 Budapest, Hungary, and HUN-REN Alfréd Rényi Institute of Mathematics, Reáltanoda utca 13--15, H-1053 Budapest, Hungary,  {\tt tobias@cs.bme.hu}}}}
\thispagestyle{empty}

\bigskip

\centerline{\small(\version)} 
\vspace{.5cm} 
 
\begin{quote} 
{\small {\bf Abstract:}} 
In this paper we propose a Moran model 
that 
\igb describes the population dynamics of \bl 
two types: While the first type has a {\em selective advantage} during reproduction, the second type can avoid \igb replacement during reproduction with some positive probability
\bl
by switching temporarily into a {\em dormant state}. 

We investigate the interplay of both evolutionary strategies by studying the invasion dynamics of the dormant type into the resident (selectively advantageous) population in the large population limit of the system. \igb It turns out that \bl the dormancy trait can not only {\em invade} and subsequently {\em fixate} under suitable parameter assumptions \igb despite its selective disadvantage \bl (a phenomenon that has already been observed in \igb a related context in \bl Blath and Tóbiás \cite{BT19}), but \igb that \bl  there is also a \igb novel \bl regime  of {\em stable \color{black} coexistence} of both types \igb due to a frequency-dependent balancing effect \bl that did not arise in the previous setup \igb with Lotka--Volterra type symmetric competition.

The emergence of a coexistence \igb regime here \bl rests \igb in part \bl on specific properties of the Moran modelling framework, in particular \igb {\em its fixed overall population size} \igb that enforces instant re-colonization after death events, \bl  as well as on the \igb (positive) \bl  mortality and resuscitation rates \bl of the dormant state.
We provide heuristic explanations for the observed types of behaviour \igb and \bl the corresponding proofs, which involve comparisons to \igb suitable \bl branching processes, approximations by dynamical systems, and an analysis of asymptotic behaviour of the latter.

\end{quote}




\section{Introduction}

Dormancy is an ubiquitous trait that allows organisms to switch into reversible and non-proliferative states that are protected from unfavourable conditions~\cite{LdHWB21}. It is generally believed that this strategy, although \igb typically \bl being costly, can  contribute to the survival and resilience of species, and may also foster coexistence and diversity in species communities. However, the precise consequences of dormancy in given scenarios are highly sensitive to the modelling of the dormancy-related mechanisms. 

For example, in {\em population genetics}, which usually assumes fixed (but large) population sizes and long time-scales, it has been shown that dormancy can affect ancestral relationships in multiple and idiosyncratic ways, \igb leading to novel coalescent structures \bl \cite{KKL01, BGKW16}, \cite{BGKW18}, 
depending on the involved time-scales and switching mechanisms (e.g.\ responsive vs stochastic). 

Moreover, in {\em (stochastic) adaptive dynamics}, where population sizes are regulated by Lotka--Volterra type competition on the ecological scale, it has been shown that so-called {\em \igb competition-induced dormancy}, even when combined with a reduced reproductive rate \igb in order to incorporate the costs of such a strategy, \bl can successfully invade and  subsequently fixate in a resident populations without this trait \cite{BT19}. A similar \igb (costly) \bl  invasion is not possible for spontaneous dormancy, see e.g.\ in \cite{P24}, \igb again highlighting the relevance of the precise dormancy strategy. This seems to be even more the case in \bl random environments, \igb where \bl several different kinds of dormancy can become optimal, depending on the distributional properties of the environment \cite{BHS21}.

Recently, \igb probabilistic \bl research is increasingly bridging the realms of both \igb evolutionary and ecological \bl modelling frameworks \igb -- note \bl  for example the identification of ancestral structures from population genetics in \igb logistic \bl  systems \cite{AFS25,F25,GN25}, or \igb the application of coalescent theory \bl in models from stochastic adaptive dynamics on a suitable time-scale \cite{CH24}. 


In the present paper, we follow \igb these research threads \bl by transferring the idea of competition-induced dormancy into the population-genetic framework of a fixed-size Moran model, \igb more precisely \bl a variant of Moran's original model~\cite{M58} without mutation and with some additional transitions, \igb all considered on the ecological scale. \bl 
Here, competitive edges become selective advantages, and killed individuals are instantaneously replaced by new offspring.  While one would at first glance expect a qualitatively similar behaviour to the one described \igb for a corresponding competitive birth--death model \bl in \cite{BT19} (invasion and subsequent fixation \igb of the dormancy trait), \bl it turns out that new effects, resting \igb crucially \bl on the specifics of the Moran model framework, can appear. In particular, invasion can now lead to \igb {\em coexistence} \bl of both types instead of fixation of the invader. The corresponding parameter regimes for all possible outcomes, including also so-called \igb {\em founder control} \bl (a term borrowed from spatial ecology, see e.g.\ \cite{V15}, referring to fixation of the resident and extinction of the mutant, regardless of their types), can be identified explicitly. This yields an example for a scenario in which  dormancy indeed fosters coexistence, \igb in line with general ecological predictions.  

\medskip

The {\bf organization of the paper} is as follows. In Section~\ref{sec-modeldefmainresultsMoran}, we define our model and state our main results. In particular, in Section~\ref{sec-modeldefMoran} we introduce our Moran model, which consists of 
two types of individuals: One having a selective advantage and the other being able to avoid selective pressure by switching to a dormant state. 
After some preliminary results on the underlying dynamical system in Section~\ref{sec-dynsystMoran} and on the branching processes approximating small sub-populations in  Section~\ref{sec-3phasesMoran}, 
in Section~\ref{sec-resultsMoran}, we present our main results, in particular Theorem~\ref{theorem-mainMoran}, which determines the fate of a single mutant individual of the dormancy type appearing in a large resident population of the other type. As part of this theorem, we also compute the asymptotic probability of a successful mutant invasion in the large-population limit, as well as the asymptotic time until reaching a macroscopic population size. We further characterize whether a successful invasion leads to coexistence vs a complete fixation of mutants and extinction of residents. 

In Section~\ref{sec-discussionMoran} we discuss our modelling choices as well as relations to \igb previous \bl work, and we interpret the conditions for invasion and coexistence \igb from the \bl main theorem. Finally, after the statement and proof of some preliminary results on the global stability properties of the underlying two-dimensional dynamical system in Section~\ref{sec-preliminaryproofsMoran}, we carry out the proof of our main result in Section~\ref{sec-proofMoran}.

\section{Model definition and main results}\label{sec-modeldefmainresultsMoran}
\subsection{Model definition}\label{sec-modeldefMoran}
\igb We consider a population of $N\in \mathbb N$ individuals which can each have two different traits called type $1$ and type $2$. Type $2$ individuals additionally come in two states, called active state (denoted by $2a$) and dormant state ($2d$), while we consider type $1$ particles to be always active. 
\bl
The three-dimensional vector of type \igb counts \bl $\mathbf X^{(N)}(t)=\big(X_{1}^{(N)}(t),X_{2a}^{(N)}(t),X_{2d}^{(N)}(t)\big)$ (\igb for e.g.\ type $1$ given by \bl $X_{1}^{(N)}(t)=\# \{ \text{individuals of type $1$ at time $t$}\}$) evolves as a continuous-time Markov chain with state space $\mathbb N_0^3$ \igb and transitions \bl as follows.

There are {\em \igb Moran type reproduction events} involving active (type $1$ and 2a) individuals. \igb Note that a particularity of the Moran dynamics is that any birth event is matched with a death event at the same time, so that the overall population size is always constant. This effectively entangles birth, death, and competition. \bl Each pair of active individuals meets at rate $\frac{1}{N}$ independently of the other pairs. 
\begin{enumerate}
    \item When two type $1$ individuals meet, a uniformly chosen one reproduces and the other one dies, and thus the type \igb counts \bl remain unchanged. 
    \item We fix $p \in (0,1)$. When two type $2$ individuals meet, we uniformly choose one of them who tries to reproduce, \igb thus exercising competitive pressure on the other. \bl With probability $1-p$, the reproduction is successful, \igb the latter individual dies and is replaced by one offspring of the former. Since both are of  the same type, the type counts \bl remain unchanged. \igb However, with \bl  probability $p$, the non-reproducing individual escapes into dormancy (switching to type $2d$), does not vanish, and thus prevents reproduction.
    \item When a type $1$ and a type $2a$ individual meet, {\igb we assume that type $1$ has a selective advantage that amounts to being more likely to be chosen as the reproducing individual. This is realized as follows. }  
    \begin{enumerate}
        \item With probability $\frac{1+s}{2}$, the type $1$ individual gets a chance to reproduce. If this happens, then with (conditional) probability $1-p$, it replaces the type $2a$ individual by its offspring, while with probability $p$, the type $2a$ individual again becomes dormant and the type $1$ individual cannot reproduce.
        \item With probability $\frac{1-s}{2}$, the type $2a$ individual reproduces and it replaces the type $1$ individual with its offspring.
    \end{enumerate} 
\end{enumerate}  
Additionally, type $2d$ individuals
\begin{enumerate}
    \item switch \igb back \bl into the active state (2a) at rate $\sigma>0$,
    \item and die at rate $\kappa$. Upon death, in line with the Moran model assumption of fixed population size, a parent is picked from the currently present population and produces one offspring of its own kind that takes the vacant slot of the deceasing individual. (In particular, if the picked parent is of type $2d$, nothing happens.)
\end{enumerate}

\begin{remark}
When the proportion of type $1$ individuals in the population is close to 1, a positive $\kappa$ is advantageous for type $1$ \igb since dying individuals in the dormant subpopulation will be replaced with high probability with type $1$ individuals -- in contrast to the case $\kappa=0$, where this is not possible. Hence, \bl when a small subpopulation of type $2$ tries to invade a type $1$ resident population, \igb a \bl $\kappa>0$ makes this invasion more difficult. 
\end{remark} \color{black}

According to the description above, if the current state of $\mathbf X^{(N)}(\cdot)$ is $(\ell,n,m)\in \mathbb N_0^3$, then the process jumps to
\begin{itemize}
    \item $(\ell+1,n-1,m)$ at rate $\frac{1}{N}\ell n\frac{1+s}{2}(1-p)$,
    \item $(\ell-1,n+1,m)$ at rate $\frac{1}{N}\ell n \frac{1-s}{2}$
    \item $(\ell,n-1,m+1)$ at rate $ \frac{1}{N}(\ell n \frac{1+s}{2}+n(n-1)) p$,
    \item $(\ell,n+1,m-1)$ at rate $\sigma m + \kappa m\frac{n}{N}$,
    \item $(\ell+1,n,m-1)$ at rate $\kappa m\frac{\ell}{N}$. 
\end{itemize}

\igb For a detailed discussion of our modelling choices we refer to \bl Section~\ref{sec-modellingchoicesMoran} below. 

\subsection{The dynamical system and its equilibria}\label{sec-dynsystMoran}

Thanks to \cite[Theorem 2.1, p.~456]{EK}, on any fixed time interval of the form $[0,T]$ \igb for some \bl $T \geq 0$, the rescaled process  $\big(\frac1N \mathbf X^{(N)}(t)\big)$ converges uniformly in probability to the solution of a certain system of ODEs, 
given convergence of initial conditions in probability. Since $X^{(N)}_1(t)+X^{(N)}_{2a}(t)+X^{(N)}_{2d}(t)=N$, this system of ODEs \igb is essentially two-dimensional. We could for example focus on the coordinates \bl
$x_{1}(t)$ and $x_{2a}(t)$, \igb with the redundant coordinate \bl $x_{2d}(t)$ \igb then determined by \bl $x_{2d}(t)=1-x_{1}(t)-x_{2a}(t)$, for $t \geq 0$. \igb In this case, the limiting \bl system reads
\[
\begin{aligned}
    \dot x_{1}(t)& = x_{1}(t)x_{2a}(t)\Big(\frac{1+s}{2}(1-p)-\frac{1-s}{2}\Big)+\kappa x_{1}(t)(1-x_1(t)-x_{2a}(t)), \\
    \dot x_{2a}(t) &= x_{1}(t)x_{2a}(t) \Big(\frac{1-s}{2} - \frac{1+s}{2}\Big) - x_{2a}(t)^2 p + \kappa x_{2a}(t)(1-x_1(t)-x_{2a}(t)) + \sigma (1-x_1(t)-x_{2a}(t)), \\
\end{aligned}
\]
\igb which after some simplification becomes \bl
\[
\begin{aligned}
      \dot x_{1}(t)& = x_{1}(t) \Big( x_{2a}(t) (s-\frac{1+s}{2}p)+\kappa (1-x_1(t)-x_{2a}(t)) \Big), \\
    \dot x_{2a}(t) &= -x_{1}(t)x_{2a}(t) s - x_{2a}(t)^2 p + (1-x_{1}(t)-x_{2a}(t))(\kappa x_{2a}(t)+\sigma),
\end{aligned} \numberthis\label{2dMoran}
\]
\igb where the \bl dynamics of $x_{2d}(t)$ can be expressed \igb as \bl
\[ \dot x_{2d}(t)=x_{1}(t)x_{2a}(t)\frac{1+s}{2}p + x_{2a}(t)^2 p -x_{2d}(t)(\kappa (x_1(t)+x_{2a}(t)) +\sigma). \numberthis\label{2dexpress} \]

In order to \igb prepare the main result on invasion and coexistence in \bl our Moran model, we now briefly discuss the existence and stability of equilibria of the system~\eqref{2dMoran}. It is easy to see that $(1,0)$ is an equilibrium of the system. We can identify another \igb monomorphic \bl equilibrium of the form $(0,\bar{x}_{2a})$. In that case, we obtain
\[ 0 = -\bar{x}_{2a}^2p + (1-\bar{x}_{2a})(\kappa \bar{x}_{2a}+\sigma),\]
and thus
\[ \bar{x}_{2a}=\frac{\kappa - \sigma + \sqrt{(\kappa-\sigma)^2+4\sigma(p+\kappa)}}{2(p+\kappa)}. \numberthis\label{x2abar,kappa>0}\]

\begin{lemma}[\igb Stability of monomorphic equilibria\bl]
    \label{lemma-simultaneousstability}
    \ 
    \begin{enumerate}[(i)]
        \item If $p<\frac{2s}{1+s}$, then $(0,\bar{x}_{2a})$ is unstable and $(1,0)$ is stable. 
        \item Let $p>\frac{2s}{1+s}$. \igb In this case, we say that the {\em basic dormancy emergence condition} is satisfied. For \bl  $D=p(s+1)-2s$ \igb we define the {\em critical dormancy mortality rate for the
        invasion of type $1$} by \bl 
        \[\tilde{\kappa}=\frac{-2D\sigma-\sqrt{4D^2\sigma^2-4(-4p+2D)\sigma D^2}}{4D-8p},
        \numberthis\label{tildekappadef}\] 
        \igb and the {\em critical dormancy mortality rate for the invasion of type $2$} by
        \[\kappa'=\frac{\sigma}{s}\Big(\frac{1+s}{2}p-s\Big).
        \numberthis\label{hatkappadef}\] 
        \bl
        Then the following assertions hold:
        \begin{align*}
            &(S_1) \quad \text{When } \kappa \in (0,\tilde{\kappa}), \: (0,\bar{x}_{2a}) \text{ is stable}, \\
            &(S_2) \quad \text{When } \kappa \in (0, \kappa'), \: (1,0) \text{ is unstable, } \\
            &(S_3) \quad \text{When } \kappa \in (\kappa', \infty), \: (1,0) \text{ is stable, } \\
            &(S_4) \quad \text{When } \kappa \in (\tilde{\kappa}, \infty), \: (0,\bar{x}_{2a}) \text{ is unstable}.
        \end{align*}
        Further, if $\sigma > \frac{2s^2}{p(1-s)}$, then $\tilde{\kappa}<\kappa'$, and in that case, both equilibria are unstable in $(\tilde{\kappa},\kappa')$, and otherwise, $\tilde{\kappa}>\kappa'$ and both equilibria are stable in $(\kappa',\tilde{\kappa})$.
    \end{enumerate} 
\end{lemma}

The proof of this lemma will be carried out in Section~\ref{sec-stability}.
Regarding the existence of a third, coordinate-wise positive \igb (polymorphic) \bl equilibrium, we have the following lemma.

\begin{lemma}[Conditions for the existence of a polymorphic equilibrium]
    \label{lemma-xtildeexistence} 
    \leavevmode
    \begin{itemize} 
        \item If 
             $\sigma>\frac{2s^2}{p(1-s)}$, \bl  
            then the dynamical system~\eqref{2dMoran} has a coordinate-wise positive equilibrium of the form $(\tilde x_1, \tilde x_{2a})$ with $\tilde x_1, \tilde x_{2a} >0$ and $\tilde x_1 + \tilde x_{2a} <1$ if and only if $\kappa \in (\tilde{\kappa},\kappa')$. If such an equilibrium exists, it is also unique and given by the formula
            \[\tilde{x}_{2a}=\frac{2\kappa s+\sigma(2s-p(1+s))}{\kappa p(s-1)+s(-2s+p(s+1))}, \qquad \tilde x_{1} = \tilde x_{2a} \big( s-\frac{1+s}{2} p - \kappa \big) \frac1\kappa + 1. \numberthis\label{thisisthecoexistenceequilibrium} \]
            \color{black}
        \item In contrast, if 
             $\sigma<\frac{2s^2}{p(1-s)}$, \bl  
            then the dynamical system~\eqref{2dMoran} has a coordinate-wise positive equilibrium of the form $(\tilde x_1, \tilde x_{2a})$ as above if and only if $\kappa \in (\kappa', \tilde{\kappa})$. If such an equilibrium exists, it is also unique and given by~\eqref{thisisthecoexistenceequilibrium}. \color{black}
    \end{itemize}
    Moreover, in both cases, for $\kappa=\tilde \kappa$, the equilibria $(0,\bar x_{2a})$ and $(\tilde x_1,\tilde x_{2a})$ formally coincide, whereas for $\kappa=\kappa'$, the equilibria $(1,0)$ and $(\tilde x_1,\tilde x_{2a})$ formally coincide.

The condition $\sigma > \frac{2s^2}{p(1-s)}$ vs.\ $\sigma <\frac{2s^2}{p(1-s)}$ will be interpreted in Remark~\ref{remark-sigma} below.
\end{lemma}
    \normalcolor



\begin{remark}

 When  $\sigma \neq \frac{2s^2}{p(1-s)}$, the equilibrium  $(\tilde x_1,\tilde x_{2a})$ defined by~\eqref{thisisthecoexistenceequilibrium}
is almost everywhere coordinate-wise nonzero. However, $(\tilde x_1,\tilde x_{2a})$ is only coordinate-wise positive (and in $(0,1)^2$) in the cases mentioned in Lemma~\ref{lemma-xtildeexistence}.
\color{black} For the latter case, we make some conjectures on the stability of $(\tilde x_1, \tilde x_{2a})$ in Conjecture~\ref{conj-Moran} below. As an illustration, in Figure~\ref{fig-equilibria} we plot the coordinates of the equilibria of~\eqref{2dMoran} as a function of $\kappa$ for two different values of $\sigma$.\end{remark}

In what follows, we refer by \emph{founder control} to the simultaneous stability of $(1,0)$ and $(0,\bar x_{2a})$ and by \emph{stable coexistence} to the simultaneous instability of these two equilibria, cf.\ also Remark~\ref{remark-nomenclature} below.
\color{black}

\begin{figure}[h!]
    \centering
    
    \begin{subfigure}[t]{0.48\textwidth}
        \centering
        \includegraphics[width=\textwidth]{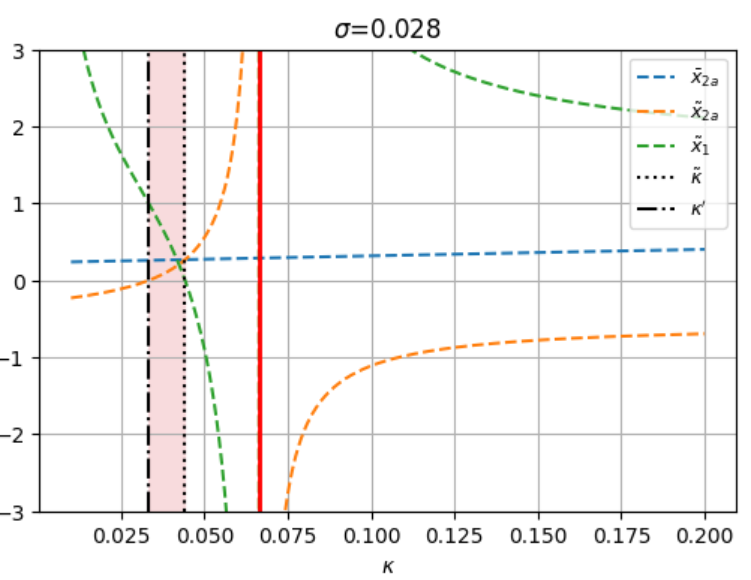}
    \end{subfigure}
    \hfill
    \begin{subfigure}[t]{0.49\textwidth}
        \centering
        \includegraphics[width=\textwidth]{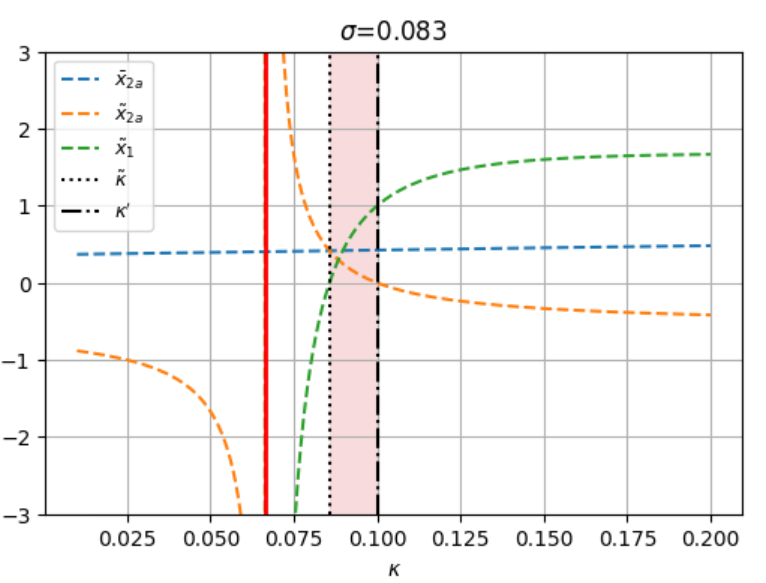}
    \end{subfigure}
    
    \caption{We plot, for $p=0.4$ and $s=0.1$, the evolution of $\bar{x}_{2a}$, $\tilde{x}_{2a}$ and $\tilde{x}_1$ according to $\kappa$, for two values of $\sigma$, the first one less \igb and the second one greater \bl than $\frac{2s^2}{p(1-s)}$.
    \igb The light-red 
    region corresponds to \bl founder control (in the first case) and coexistence (on the second) coinciding with the existence of $(\tilde{x}_{2a},\tilde{x}_{2d})$. 
   Moreover, the pictures verify the properties of $(\tilde{x}_{1},\tilde{x}_{2a})$ stated in Lemma ~\ref{lemma-xtildeexistence}. We also remark the apparition of an asymptote (red vertical line) that will be made explicit in the proof of Lemma~\ref{lemma-xtildeexistence} in Section~\ref{sec-stability}.}\label{fig-equilibria}
\end{figure}

\subsection{Three phases of an invasion and two branching processes}
\label{sec-3phasesMoran}

We will study the situation where at time zero, the initial condition of our Moran model is $\mathbf X^{(N)}(0)=(N-1,1,0)$, i.e.\ almost all individuals are from the resident type $1$ and there is one single \igb active \bl type $2$ ``mutant'' individual. 
Our main result below, Theorem~\ref{theorem-mainMoran}, \igb determines \bl 
the fate of the population after the appearance of this first mutant individual for large $N$. \igb To be able to state and prove \bl 
this theorem, \igb we need to define the (asymptotic) {\em survival probability} of an invader, and its {\em invasion fitness}. Both quantities can be derived from a suitable coupled branching process for the invader during the early phase of the invasion. Hence, we already here discuss the three classical phases of the invasion, similarly to the \igb strategy in the \bl seminal paper of Champagnat~\cite{C06}, where a competitive Lotka--Volterra type model with logistic competition was studied.
\begin{itemize}
    \item {\igb \bf Phase I:} As long as the mutant population is small compared to $N$ \igb (but not extinct), \bl the type $1$ population size \igb remains \bl close to $N$ thanks to the constant \color{black} population size \igb assumption. Hence, the influence of the resident on the invader is almost constant, and we \bl can approximate the mutant population size by \igb an autonomous \bl two-type binary branching process.
    \item {\igb \bf Phase II:} This phase \igb begins once \bl the total mutant population size has successfully reached $\eps N$ for some $\eps>0$ small but not depending on $N$ \igb (thus exiting from phase I), \bl which will have an asymptotically positive probability as $N\to\infty$ \igb only if \bl $(1,0)$ is an unstable equilibrium of the dynamical system~\eqref{2dMoran}. Given this, $\frac1N \mathbf X^{(N)}(\cdot)$ can now be approximated by the solution to~\eqref{2dMoran} (where we approximate $\frac1N X^{(N)}_{2d}(\cdot)$ by $n_{2d}(\cdot)=1-n_{1}(\cdot)-n_{2a}(\cdot)$). \igb The solution then \bl
    converges to $(0,\bar x_{2a})$ whenever this equilibrium is asymptotically stable, and in case it is unstable, it \igb will approach \bl $(\tilde x_1, \tilde x_{2a})$. 
    \item {\igb \bf Phase III:} This phase only occurs \igb if \bl the dynamical system tends to $(0,\bar x_{2a})$. In this last phase, $\frac1N \big(X_{2a}^{(N)}(t),X_{2d}^{(N)}(t)\big)$ stays close to $(\bar x_{2a},\bar x_{2d})$ where $\bar x_{2d}=1-\bar x_{2a}$, while the resident population goes extinct.
\end{itemize}

\igb We now collect some necessary properties of the branching process approximation of the invader during phase I. Here, as \bl long as the \igb total type $2$ \bl population is small but not extinct, the effect of dormancy due to self-interaction of type $2a$ can be ignored and type $2d$ individuals mostly get replaced by type $1$ individuals upon death (for $\kappa>0$), and hence we can approximate the process $\big(X_{2a}^{(N)}(t),X_{2d}^{(N)}(t)\big)$ by {\igb an autonomous (i.e.\ independent of type 1)\bl} branching process $\big(Z_{2a}^{(N)}(t),Z_{2d}^{(N)}(t)\big)$ that jumps from $(i,j) \in \mathbb N_0^2$ to: 
\begin{itemize}
    \item $(i-1,j)$ at rate $i\frac{1+s}{2}(1-p)$,
    \item $(i+1,j)$ at rate $i \frac{1-s}{2}$
    \item $(i-1,j+1)$ at rate $pi\frac{1+s}{2}$,
    \item $(i+1,j-1)$ at rate $\sigma j$ ,
    \item $(i,j-1)$ at rate $\kappa j$.
\end{itemize}
{\igb This approximation will be made precise in Section~\ref{sec-proofMoran}.}

{\igb We first compute the mean growth rate of the invading population, that is, the {\em invasion fitness} of the mutant, from this branching process. Indeed, the} mean matrix of the two-type branching process is
\[ 
J:=\begin{pmatrix}
    -s & \frac{1+s}{2} p \\
     \sigma & -\sigma  - \kappa 
\end{pmatrix}, 
\]
which has an eigenvalue with positive real part if and only if its determinant is negative, i.e.\ if 
\[ \kappa<\kappa', \numberthis\label{2invades1Moran} \] 
which happens to be the condition $(S_2)$ of instability of $(1,0)$ from Lemma~\ref{lemma-simultaneousstability}, see Section~\ref{sec-invasionfitnesses} for an interpretation. When there is such an eigenvalue, we will denote it by $\lambda_2$, \igb and this is the desired mean growth rate. \bl  Solving the characteristic equation, we obtain
\[\lambda_2=\frac{-s-\sigma-\kappa+\sqrt{(s+\sigma+\kappa)^2-4(s(\sigma+\kappa) -\frac{1+s}{2}p\sigma)}}{2} \color{black}. \numberthis\label{lambda2defMoran} \]

{\igb The second quantity required for our invasion theorem is the {\em survival} (or, equivalently, the {\em extinction}) {\em probability} of the invader, again obtained from this branching process. 
More precisely, denote the latter\bl} by $q_{2a}$ \igb for the \bl branching process started from one active (type $2a$) individual and zero dormant (type $2d$) individuals; then we have $q_{2a}=1$ in the subcritical case and $q_{2a}<1$ in the supercritical case.

Using \cite[Chapter V.3]{AN72}, the pair of extinction probabilities $(q_{2a},q_{2d})$ \igb satisfies \bl
\begin{align*}
    q_{2a}&=\frac{1+s}{2}(1-p)+\frac{1-s}{2}q_{2a}^2+p\frac{1+s}{2}q_{2d}, \\
    q_{2d}&=\frac{\sigma}{\sigma+\kappa}q_{2a}+\frac{\kappa}{\sigma+\kappa}.
\end{align*}
\igb Solving this for $q_{2a}$ leads to the quadratic equation \bl
\[0=\frac{1-s}{2}q_{2a}^2+\Big(p\frac{1+s}{2}\frac{\sigma}{\sigma+\kappa}-1\Big)q_{2a}+\frac{1+s}{2}\Big((1-p)+p\frac{\kappa}{\sigma+\kappa}\Big).\]

This second degree polynomial has a positive leading coefficient, is positive at $0$ and is equal to $0$ at $1$, moreover, \igb again \bl thanks to \cite[Chapter 5]{AN72}, it has a root in $(0,1)$ if the following matrix has an eigenvalue $>1$: 
\[ 
M:=\begin{pmatrix}
    1-s & \frac{1+s}{2} p \\
     \frac{\sigma}{\sigma+\kappa} & 0 
\end{pmatrix}. 
\]
As the characteristic polynomial is $\lambda^2-\lambda(1-s)-p\frac{1+s}{2}\frac{\sigma}{\sigma+\kappa}$, which at $1$ is equal to $s-p\frac{1+s}{2}\frac{\sigma}{\sigma+\kappa}$, and this quantity is precisely negative thanks to condition ~\eqref{2invades1Moran}, we have a root in $(0,1)$ for the polynomial function in $q_{2a}$.


In Phase III, when $\frac1N \big(X_{2a}^{(N)}(t),X_{2d}^{(N)}(t)\big) \approx (\bar x_{2a},\bar x_{2d})$ and $X_1^{(N)}(t)$ is already small compared to $N$ but not yet extinct, the type $1$ population size $X_{1}^{(N)}(t)$ can be approximated by the branching process $Z_1(t)$ that jumps from $n \in \N_0$ to
\begin{itemize}
    \item $n+1$ at rate $n \bar x_{2a} \frac{1+s}{2}(1-p) + \kappa \bar x_{2d}$,
    \item $n-1$ at rate $n \bar x_{2a} \frac{1-s}{2}$.
\end{itemize}
The net growth rate (i.e.\ birth rate minus death rate) of the branching process corresponding to the invasion of type $1$ against type $2$ is 
\[ -\lambda_1:=-\Big(\bar x_{2a}( \frac{1+s}{2}p-s) -\kappa\bar{x}_{2d}\Big) \color{black}, \numberthis\label{lambda1defMoran} \] i.e., when $p>\frac{2s}{1+s}$, \color{black} this branching process is supercritical if and only if 
\[ \bar{x}_{2a}<\frac{\kappa}{\kappa+\frac{1+s}{2}p-s}, \numberthis\label{phase3cond} \] which happens to be the condition for instability of $(0,\bar{x}_{2a})$ in that case, as we will see later in the proofs, 
and subcritical otherwise. When $p<\frac{2s}{1+s}$, the branching process is always supercritical.


\begin{remark}[\igb Reverse invasion of type 1 and fixation of type 2\bl]
    \label{remark-nomenclature}
    One \igb can \bl also study the reverse invasion direction, where a small population of type 1 tries to invade a resident population of type 2 with an initial condition such that $X_1^{(N)}(0) = 1$ and $\frac1N \big(X_{2a}^{(N)}(0),X_{2d}^{(N)}(0)\big)\approx (\bar x_{2a}, \bar x_{2d})$. The branching process approximating the type 1 population in Phase I of this invasion is the same as the one that approximates the type 1 population in Phase III of the invasion of type 2, and the same holds with the roles of type 1 and 2 swapped. We will formally not study the invasion of type 1 in this paper, but at least on a heuristic level, it is clear that type 1 can invade if and only if type 2 cannot fixate, while type 1 fixates if and only if type 2 cannot invade. As anticipated, the case where none (resp.\ both) of the types can invade the other one turns out to be equivalent to founder control (resp.\ stable coexistence). 
\end{remark}

For additional interpretation and equivalent conditions for the sub- and supercriticality of the branching processes and the emergence of the coordinate-wise positive equilibrium $(\tilde x_{2a}, \tilde x_{2d})$, see Section~\ref{sec-invasionfitnesses}.

\subsection{Statement of the main results}\label{sec-resultsMoran}
\igb To identify the three phases of the potential invasion and the corresponding extinction and coexistence scenarios, we define several hitting times. Indeed, 
for \bl $\eps \geq 0$ \igb let \bl 
\[ T_{\eps}^{1} = \inf \{ t \geq 0 \colon X_1^{(N)}(t)=\lfloor \eps N \rfloor \}, \]
the first time when the type $1$ population size reaches size $\lfloor \eps N \rfloor$. In particular, $T_0^1$ is the extinction time of type $1$. Certainly, $T_\eps^1$ also depends on $N$, but we suppress this in the notation for simplicity. Define further $X_{2}(t)=X_{2a}(t)+X_{2d}(t)$ for $t \geq 0$, and for $\eps \geq 0$ put
\[ T_{\eps}^2 = \inf \{ t \geq 0 \colon X_2^{(N)}(t) = \lfloor \eps N \rfloor \}. \numberthis\label{Teps2defMoran} \]
Moreover, when the coordinate-wise positive \igb (polymorphic) \bl equilibrium $(\tilde x_1,\tilde x_{2a})$ of the dynamical system~\eqref{2dMoran} exists, for $\beta>0$ we define the \emph{coexistence region}
\[ A_{\beta}^{\rm co} := (\tilde x_1-\beta,\tilde x_1+\beta) \times (\tilde x_{2a}-\beta, \tilde x_{2a}+\beta) \times (\tilde x_{2d}-\beta,\tilde x_{2d}+\beta) \]
and the \igb first entrance time in the coexistence region, \bl 
\[ T_{\beta}^{\rm co} = \inf \big\{ t\geq 0 \colon \frac1N \big(X_1^{(N)}(t),X_{2a}^{(N)}(t),X_{2d}^{(N)}(t)\big) \in A_{\beta}^{\rm co} \big\},  \]
where $\tilde x_{2d} = 1-\tilde x_1-\tilde x_{2a}$. When this equilibrium does not exist, we define $A_{\beta}^{\rm co} = \emptyset$ and consequently $ T_{\beta}^{\rm co} =+\infty$. 

\begin{defn}[Invasion scenarios]
\label{de:scenarios}
For our main theorem, we distinguish the following scenarios:
\begin{enumerate}[(i)]
\item  \emph{Type 1 (resident) fixates:} $\lambda_2<0$ and $\lambda_1 < 0$.
\item  \emph{Invasion and fixation of type 2:}  $\lambda_2>0$, $\lambda_1 >0$, and
conditional on the event $\{ T_0^2 > T_0^1  \}$, we have
\[ \lim_{N\to\infty} \frac{T_0^1}{\log N} = \frac{1}{\lambda_2}+\frac{1}{\lambda_1} \qquad \text{in probability} \numberthis\label{fixationtimeMoran} \]
\igb as well as \bl
\igb
\[ \lim_{N\to\infty} \frac1N \big(X_{2a}^{(N)}(T_0^1),X_{2d}^{(N)}(T_0^1)\big) = (\bar x_{2a}, \bar x_{2d} ) \qquad \text{in probability}. \numberthis\label{T01proportions} \]
\item  \emph{Coexistence:} $\lambda_2>0$, $\lambda_1<0$, and for all sufficiently small $\beta>0$, 
conditional on the event $\{ T_0^2 > T_0^1 \wedge T_{\beta}^{\rm co} \}$, we have
\[ \lim_{N\to\infty} \frac{T_\beta^{\rm co}}{\log N} = \frac{1}{\lambda_2} \qquad \text{in probability}. \numberthis\label{coextimeMoran}\]
\item \emph{Founder control:} $\lambda_2<0$ and $\lambda_1 > 0$. 
\end{enumerate}
\end{defn}
\bl


\begin{remark}
Recall from Section~\ref{sec-3phasesMoran} that $\lambda_2>0$ is equivalent to $q_{2a}<1$, i.e.\ the non-extinction of the mutant type 2 with asymptotically positive  probability. Hence, conditioning on $T_0^2$ having a large value is only relevant in the case of invasion and fixation of type 2 (ii) and coexistence (iii). Moreover, it follows from Lemma~\ref{lemma-xtildeexistence} that apart from the critical case $\sigma = \frac{2s^2}{p(1-s)}$, the coordinate-wise positive equilibrium $(\tilde x_1,\tilde x_{2a})$ exists if and only if $\lambda_1$ and $\lambda_2$ have opposite signs, i.e.\ precisely in the case of coexistence (iii) and founder control (iv). This way, replacing $\{ T_0^2 > T_0^1 \}$ with $\{ T_0^2 > T_0^1 \wedge T_\beta^{\rm co} \}$ for $\beta>0$ would make no difference in case (ii).
\end{remark}
\color{black}

Then our main result is the following:

\begin{theorem}[\igb Invasion, fixation, coexistence and founder control\bl]
\label{theorem-mainMoran}
Assume that $p \neq \frac{2s}{1+s}$, 
$\kappa \notin \{ \tilde{\kappa},\kappa' \}$,  and $\mathbf X^{(N)}(0)=(N-1,1,0)$. Then for all sufficiently small $\beta>0$ we have
\[ \lim_{N\to\infty} \P \big( T_0^2 < T_0^1 \wedge T_{\beta}^{\rm co} \big) = q_{2a} \]
and conditional on the event $\{ T_0^2 < T_0^1 \wedge T_{\beta}^{\rm co}  \}$, we have
\[ \lim_{N\to\infty} \frac{T_0^2}{\log N} =0 \qquad \text{in probability}. \numberthis\label{sublogMoran} \]
Moreover,
\begin{enumerate}[(a)]
    \item if $p<\frac{2s}{1+s}$, then (i) holds (1 fixates) for all $\kappa \geq 0$.
    \item if $p>\frac{2s}{1+s}$ and $\sigma > \frac{2s^2}{p(1-s)}$, then 
    \medskip
    \igb
    \begin{itemize}
        \item (i) holds (1 fixates) for $\kappa > \kappa'$, 
        \item (ii) holds (2 fixates) for $0 \leq \kappa <\tilde\kappa$,
        \item and (iii) (coexistence) holds for $\tilde \kappa < \kappa < \kappa'$.
    \end{itemize}
    \bl
    \medskip
    \item if $p>\frac{2s}{1+s}$ and $\sigma < \frac{2s^2}{p(1-s)}$, then 
    \medskip
    \igb
    \begin{itemize}
        \item (i) holds for  $\kappa > \tilde\kappa$,
        \item (ii) holds (2 fixates) for $0 \leq \kappa <  \kappa'$,
        \item and (iv) (founder control) holds for $\kappa \in ( \kappa',\tilde\kappa)$.
    \end{itemize}
    \bl
\end{enumerate}

\end{theorem}

\igb See \bl Figures~\ref{fig:coexistence_range}, \ref{fig:no_coexistence_range}, and \ref{fig-regimes} for \igb a vizualisation of the corresponding parameter regimes.\bl

\begin{figure}[h!]
    \begin{center}
        \begin{tikzpicture}
            \draw[|-|, orange] (7.3,1.3) -- (1,1.3) node[above,yshift=0.3cm, right] {$(0, \bar x_{2a})$ stable};         
            \draw[|-, blue] (3.7,1.9) -- (10,1.9)node[above,yshift=0.3cm, left] {$(1, 0)$ stable};    
            \fill[pattern=checkerboard, pattern color=purple!20] (3.7,0) rectangle (7.3,1);
            \node at (5.5,0.5) {founder control};
            \fill[red!10] (1,0) rectangle (3.7,1);
            \node at (2.5,0.5) {$2$ fixates};
            \fill[blue!20] (7.3,0) rectangle (10,1);
            \node at (8.5,0.5) {$1$ fixates};
            \node at (8.5,0.5) {$1$ fixates};
            \draw[purple, thick, dashed] (3.7,0.5) -- (3.7,-0.5) node[below] {$ \kappa'$};
            \draw[blue, thick, dashed] (7.3,0.5) -- (7.3,-0.5) node[below] {$\tilde \kappa$};
            \draw[black, thick, dashed] (1,0.5) -- (1,-0.5) node[below] {$0$};
            \draw[->, thick] (0.7,0) -- (10.5,0) node[below, right] {$\kappa$};   
        \end{tikzpicture}
        \caption{In the case $p>\frac{2s}{1+s}$ and $\sigma < \frac{2s^2}{p(1-s)}$, we have $\tilde \kappa >\kappa'$. In the founder control regime where $\kappa \in ( \kappa', \tilde \kappa)$, both equilibria $(0,\bar x_{2a})$ and $(1,0)$ are locally asymptotically stable, and the coordinate-wise positive equilibrium $(\tilde x_1, \tilde x_{2a})$ also exists (but we expect it to be unstable).}\label{fig:coexistence_range}
        \end{center}
    \end{figure}
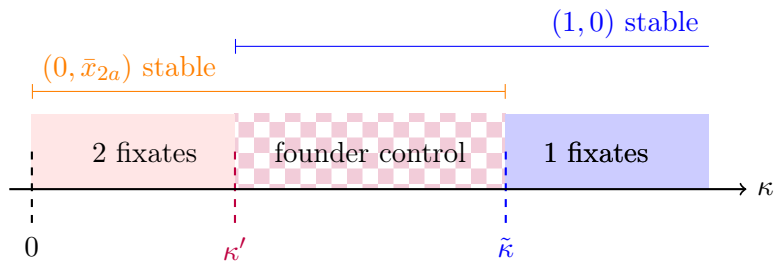

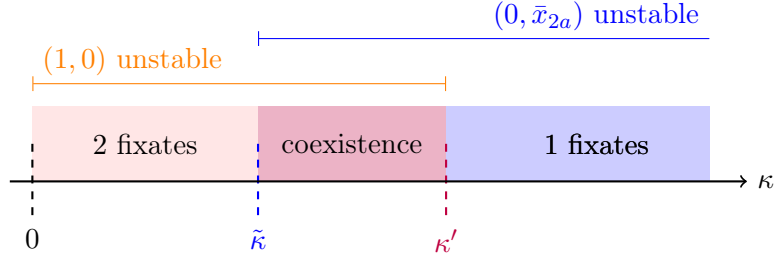
\begin{figure}[h!]
    \begin{center}
        \begin{tikzpicture}
            \draw[|-|, orange] (6.5,1.3) -- (1,1.3) node[above,yshift=0.3cm, right] {$(1,0)$  unstable};          
            \draw[|-, blue] (4,1.9) -- (10,1.9) node[above,yshift=0.3cm, left] {$(0, \bar x_{2a})$ unstable};    
            \fill[color=purple!30] (4,0) rectangle (7,1);
            \node at (5.25,0.5) {coexistence};
            \fill[red!10] (1,0) rectangle (4,1);
            \node at (2.5,0.5) {$2$ fixates};
            \fill[blue!20] (6.5,0) rectangle (10,1);
            \node at (8.5,0.5) {$1$ fixates};
            \node at (8.5,0.5) {$1$ fixates};
            \draw[blue, thick, dashed] (4,0.5) -- (4,-0.5) node[below] {$\tilde \kappa$};
            \draw[purple, thick, dashed] (6.5,0.5) -- (6.5,-0.5) node[below] {$\kappa'$};
            \draw[black, thick, dashed] (1,0.5) -- (1,-0.5) node[below] {$0$};
            \draw[->, thick] (0.7,0) -- (10.5,0) node[below, right] {$\kappa$}; 
        \end{tikzpicture}
        \caption{In the case $p>\frac{2s}{1+s}$ and $\sigma > \frac{2s^2}{p(1-s)}$, we have $\tilde \kappa < \kappa'$. In the coexistence regime where $\kappa \in (\tilde \kappa, \kappa')$, both equilibria $(0,\bar x_{2a})$ and $(1,0)$ are unstable, and the coordinate-wise positive equilibrium $(\tilde x_1, \tilde x_{2a})$ emerges (and we expect it to be asymptotically stable, see Conjecture~\ref{conj-Moran}).}\label{fig:no_coexistence_range}
        \end{center}
    \end{figure}

The proof of Theorem~\ref{theorem-mainMoran} will be carried out in Section~\ref{sec-proofMoran}. \igb It will be \bl based on the proofs of the previous lemmas on the existence and stability of equilibria of the dynamical system~\eqref{2dMoran} (see Section~\ref{sec-stability}), as well as on some additional results on the convergence of the dynamical system, which will be presented in Section~\ref{sec-convergenceMoran} together with their proofs. As already indicated, despite the fact that the population size of our Moran model is not governed by logistic competition but constant, the main structure of the invasion analysis follows the classical three-phase paradigm originating from~\cite{C06}. We note that the approximation by branching processes and logistic curves is also reminiscent of the analysis of selective sweeps in the population-genetic literature, in particular in the context of genetic hitchhiking (see e.g.\ \cite{EPW} and the references therein). However, there are also some differences. For example, since type 2 has two sub-types in our model, we need to employ Freidlin--Wentzell type large-deviation arguments to conclude that the type 2a (resp.\ 2d) population size scaled by $N$ stays close to $\bar x_{2a}$ (resp.\ $\bar x_{2d}$) during Phase III. This is done most conveniently using the methods of~\cite{C06}. \color{black} 

New difficulties arise for the behaviour of the dynamical system~\eqref{2dMoran} in the case \bl $\kappa>0$. \igb While for $\kappa=0$ \bl the proof of Theorem~\ref{theorem-mainMoran} can almost be seen as a simplified version of the one of the main results of~\cite{BT19} (which in turn uses various proof techniques from~\cite{C+19}) up to moderate modifications accounting for the constant population size. In contrast, for $\kappa>0$ the analysis of the equilibria of the system~\eqref{2dMoran} and their stability becomes surprisingly tedious, and the proof of the convergence of the dynamical system to the desired equilibrium also requires some additional arguments such as the Bendixson--Dulac criterion. 

\begin{remark}[\igb General comments on the assertions of Theorem~\ref{theorem-mainMoran}\bl]
\
\begin{enumerate}
\item
Note that the behaviour of the dynamical system becomes richer for $\kappa>0$. Indeed, for $\kappa=0$, the case (iii) of coexistence and the one (iv) of founder control never occur. Indeed, in Figure~\ref{fig-regimes}, we see that the width of the blue founder control region (iv) vanishes as $\sigma \downarrow 0$. 
\item
Note that there is no result in Theorem~\ref{theorem-mainMoran} corresponding to~\eqref{T01proportions} about $X_{1}^{(N)}(T_0^2)$ in case (i) because $X_{1}^{(N)}(T_0^2)=N$ almost surely on the event $\{ T_0^2<\infty \}$. 
\item
Note also the additional term $+\frac{1}{\lambda_1}$ in~\eqref{fixationtimeMoran} in case (ii) compared to \eqref{coextimeMoran} in case (iii). This is the duration of the last phase of the invasion of type $2$ where the former resident (type $1$) population goes extinct.
\item 
When $\sigma = \frac{2s^2}{p(1-s)}$, $\tilde\kappa$ and $ \kappa'$ coincide. Here, we do not have any assertion about the critical case $\kappa = \tilde \kappa = \kappa'$. For $\kappa$ smaller than this expression, type $2$ fixates, and for $\kappa$ larger than it, type $1$ fixates. The proof of Theorem~\ref{theorem-mainMoran} applies in this case for $\kappa \neq \tilde \kappa = \kappa'$, but for simplicity we did not include this as a separate case in the theorem.
\end{enumerate}
\end{remark}

\begin{figure}
\includegraphics[scale=3]{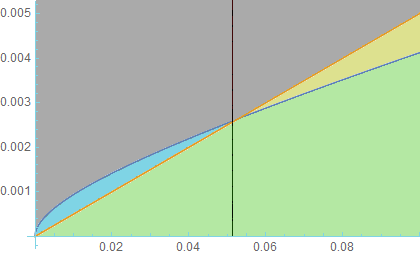}
\caption{For fixed values of $p=0.1$, and $s=0.05$ (so that $p > \frac{2s}{1+s}$), we plot $\tilde\kappa$ (blue) and $ \kappa'$ as functions of $\sigma$. Light green regime: fixation of type $2$; blue regime: founder control; gray regime: fixation of type $1$; yellow regime: stable \color{black} coexistence. The vertical line $\{ \sigma = \frac{2s^2}{p(1-s)}\}$ separates the cases (b) (left) and (c) (right). The gray regime corresponds to scenario (i), the light green one to scenario (ii), the yellow one to scenario (iii), and the blue one to scenario
(iv).
}\label{fig-regimes}
\end{figure}

\begin{remark}[Stability of the polymorphic equilibrium]
\label{remark-maintheorem}
According to Lemma~\ref{lemma-simultaneousstability}, 
case (iii) corresponds to the case when $(1,0)$ and $(0,\bar x_{2a})$ are both unstable. In this case, $(\tilde x_1,\tilde x_{2a})$ is coordinate-wise positive and we expect that it is asymptotically stable, however, this seems to be tedious to show directly via linearization.  What we actually show below in the proof of Lemma~\ref{lemma-convtoaneq2} is that in this case, solutions $(x_1(t),x_{2a}(t))$ to the dynamical system~\eqref{2dMoran} satisfying $x_1(0),x_{2a}(0)>0$ and $x_1(0)+x_{2a}(0)<1$ cannot converge to a periodic orbit as $t\to\infty$, neither to $(1,0)$ or $(0,\bar x_{2a})$, not even along a diverging sequence of times. Consequently, using the Poincaré--Bendixson theorem, the $\omega$-limit set of any solution to the dynamical system~\eqref{2dMoran} with an initial condition satisfying $x_1(0),x_{2a}(0)>0$ and $x_1(0)+x_{2a}(0)<1$ has to contain the equilibrium $(\tilde x_1,\tilde x_{2a})$. We expect but cannot prove that the $\omega$-limit set cannot be  a homoclinic orbit or a composition of multiple homoclinic orbits around this equilibrium, but just a singleton consisting of the equilibrium itself. (Note that a homoclinic orbit around $(\tilde x_1, \tilde x_{2a})$ could only exist if this equilibrium were a saddle point, which would, however, contradict Conjecture~\ref{conj-Moran}.)

In Section~\ref{sec-relationtoexistingmodels} we briefly compare our results on coexistence, fixation of one type, and founder control to some prior results in related stochastic population-dynamic models (partially with dormancy).
\end{remark}

\begin{figure}
\includegraphics[scale=1]{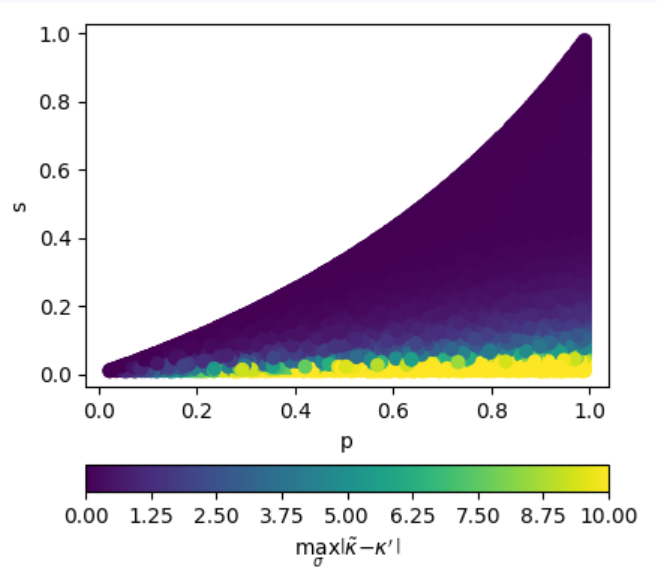}
\caption{We plot, for $p$ and $s$ satisfying $p>\frac{2s}{1+s}$, the maximum of the size of the interval where we have coexistence or founder control according to $\sigma$. Here, we generated a large number of random pairs $(s,p)$ such that $p>\frac{2s}{1+s}$  and then chose the maximizer of $|\tilde\kappa-\kappa'|$ from a finite set of values of $\sigma \in (0,1)$.}\label{fig-coexistence_interval_size}
\end{figure}

\begin{remark}[Dependence of the size of the coexistence/founder control region on $s$ and $p$\bl]
    The dependence of the size of coexistence or founder control regions on the underlying parameters is rather complex. While Figures \ref{fig:coexistence_range} and \ref{fig:no_coexistence_range} are purely schematic, and Figure \ref{fig-regimes} is for fixed $s$ and $p$, we conclude this section with a quick glance at the dependence of this region on $s$ and $p$. Indeed, as \bl we can see in Figure~\ref{fig-coexistence_interval_size}, the size of the coexistence/founder control interval 
    becomes \bl larger for small values of $s$ (yellow region), except when $p$ is also small. \bl Indeed, $\kappa'$ tends to infinity when $s$ tends to $0$ (unless it is compensated by simultaneously taking $p$ to $0$), while $\tilde{\kappa}$ converges to a positive value for $s$ to $0$ (cf.\ its representation in Lemma \ref{lemma-simultaneousstability}), so that we obtain an increasing coexistence region for small $s$ (and all other parameters fixed). (This will indeed be a coexistence and not a founder control region for small $s$ because $\kappa'$ will be larger than $\tilde\kappa$.) \color{black}
    \bl
    Intuitively, \bl for decreasing \bl $s$ the residents lose their reproductive \bl advantage, and need large values of $\kappa$ to compensate this loss from the high mortality in the seed bank. \bl This follows from the fact that the residents fixate for $\kappa>\kappa'$. Putting it differently, \color{black} for fixed $\kappa$ in the invasion regime of residents (i.e.\ the union of the regime of fixation of type 1 and the one of coexistence)\color{black}, the fixation of residents will be replaced by a coexistence regime when $s$ tends to 0. 
\end{remark}

\section{Discussion}\label{sec-discussionMoran}

In this section, we \igb begin with a \bl review our modelling assumptions in Section \ref{sec-modellingchoicesMoran}, discuss the \igb particular \bl consequences of the fixed population size \igb assumption of \bl the Moran model in Section~\ref{sec-consequences}, \igb review \bl some relations to existing work in Section \ref{sec-relationtoexistingmodels}, interpret the mutual invasion conditions in Section \ref{sec-invasionfitnesses} and finally comment on the metastability of coexistence equilibria for stochastic systems in Section \ref{sec-metastability}.

\subsection{Modelling choices}
\label{sec-modellingchoicesMoran}

In our Moran model, reproduction is incorporated by picking pairs of particles at a fixed rate, and choosing one of them for (potential) death and the other one for potential reproduction. This is formally  similar to modelling competition events in the approximation of Lotka--Volterra models (as in \cite{FM04}), where, however, only one particle dies, and reproduction \igb is \bl completely disentagled from competitive events. We include a selective advantage of one type with the help of our parameter $s$, which gives a higher chance to the type $1$ particle to become the reproducing parent. Note that the fixed population size assumption of the Moran model here enforces the \igb simultaneity \bl  of competitive death and reproduction.

One novelty in this Moran set-up is that particles exhibiting a dormancy trait can avoid competitive death and thus compensate selective disadvantages in a particular way, similar to the Lotka--Volterra set-up in \cite{BT19}. We deem it  natural to allow escape into dormancy both for intra- and inter-species competition (with the same probability).

A similar entanglement appears in the dormant state \igb for strictly positive death rates $\kappa$. \bl Here, if a dormant particle dies, we assume that its spot is immediately taken by a uniformly chosen particle from the \igb whole (currently present) \bl population. Since we do not consider this a competitive/selective event, individuals of \color{black} all types are equally likely to recolonize. Again, this instant recolonization is due to the fixed population size assumption of the Moran. Interestingly, this is \igb also precisely the mechanism that allows for the emergence of coexistence (for suitable parameters), \bl in contrast to the Lotka--Volterra type model in \cite{BT19}, where dormant individuals died at a linear rate without replacement and coexistence (or founder control) was impossible. See the discussion in Section~\ref{sec-consequences} below.

In the light of the above, two modification of the dynamics of the model seem to be \igb plausible, \bl although \igb we refrain \bl from investigating their full consequences here:
\begin{itemize}
    \item One could introduce a selective advantage according to $s$ \igb also \bl when re-colonizing vacant slot due to death in the dormant state. 
    \item One could pick different different transition probabilities, \igb say \bl $p_1, p_{2a}$, for escape into dormancy depending on the type of the potential parent ($1$ or $2a$), and even $p_{2a}=0$ might be \igb considered. \bl
\end{itemize}
We believe that in both cases, coexistence regimes still exist, but \igb do not delve into \bl  the details here.

Regarding time-scales, recall that in our model we assume strong selection \igb of order \bl $O(1)$ and weak dormancy times (also $O(1)$) on the ecological scale. One might be tempted to switch  to an evolutionary scale $O(N)$ with weak selection $s/N$ and strong dormancy times $O(N)$ and try to identify a corresponding diffusion limit. We leave this approach for future investigation.

\subsection{Consequences of the Moran modelling framework}
\label{sec-consequences}
It turns out that the {\em fixed population size assumption} of the Moran model has profound consequences, both on the level of results and the technical treatment of the model. These stem from two main effects: First, the death of an individual always occurs simultaneously with birth of another one, so that there are no `vacant population slots' (instant recolonization). Second, the fixed total population size means that in the description of the model one type out of $1, 2a, 2d$ is always redundant, and the corresponding ODE limits are two- instead of three dimensional. 


We begin with the consequences of instant recolonization, in particular during the death of a dormant individual. In absence of this, i.e.\ for $\kappa=0$, invasion implies fixation: Coexistence and founder control never occur. This was also already true in the Lotka--Volterra type model of~\cite{BT19} with symmetric competition and competition-induced dormancy, even in the variant of the model that included death in the dormant state. While for $\kappa=0$, the proof of Theorem~\ref{theorem-mainMoran} is an easy adaptation of that of the main result of~\cite{BT19}, this is not true anymore for $\kappa>0$ and in that case, the behaviour of the model of the present paper becomes much more complex.  

The point is that recolonization of slots vacated by type 2 particles always helps the type 1 particles. However, it turns out that the effect of increasing $\kappa$ is different for small $\sigma$ than for large $\sigma$. Assume that the choice of the parameters $s,p,\sigma$ is such that for $\kappa=0$, type 2 fixates. Then, if $\sigma$ is small, increasing $\kappa$ first makes $(1,0)$ stable and thus it leads to founder control (lack of invasion of type 2), and only afterwards it makes $(0,\bar x_{2a})$ unstable, leading to the fixation of type 1. When $\sigma$ is large, the stability of these two equilibria changes in the reverse order, and in the intermediate regime where $(0,\bar x_{2a})$ is already unstable but $(1,0)$ is still also unstable, we have coexistence (invasion of type 1). Eventually, for $\kappa$ large enough, type 1 will again reach fixation.

The lower dimensionality, in particular for the many-particle ODE limit, has technical advantages, allowing for a more complete analysis of asymptotic behaviour involving elements of the Poincaré--Bendixson theory, in particular the Bendixson--Dulac criterion. However, on a technical level, computations related to the existence and stability of equilibria of the dynamical system are often still substantially more involved than in the case of the dynamical system corresponding to~\cite{BT19}, despite the fact that the latter system is rank-3 three-dimensional.

\subsection{Relation to existing models}\label{sec-relationtoexistingmodels} The dichotomy between a coexistence regime and a founder control regime has been observed in many different types of population-dynamic models before. One of the most classical population-dynamic models is the two-type competitive Lotka--Volterra model. Here, when competition is {\em symmetric} between the two types, invasion implies fixation (namely, the type with the higher net reproduction rate wins). But if one makes competition {\em asymmetric}, depending on the strength of intra- vs.\ interspecific competition, a regime of stable coexistence or one of founder control emerges. In the regime of stable coexistence, the system has a stable coordinate-wise positive equilibrium, while in the case of founder control, it has an unstable one, and in the case of fixation of one type it has no coordinate-wise positive equilibrium. See e.g.\ \cite{Bomze} for a reference. 

In these Lotka--Volterra models, coexistence is equivalent to the coordinate-wise nonzero equilibrium having two positive coordinates and being stable\color{black}, which also both have to be smaller than the corresponding positive coordinates of the equilibria with exactly one positive and one zero coordinate. The characterization of founder control is similar, but here the coordinate-wise positive equilibrium is unstable. \color{black} This is also similar to our model because here, for $\sigma > \frac{2s^2}{p(1-s)}$, type 1 can invade if and only if $\tilde x_{2a} < \bar x_{2a}$ and type 2 can invade if and only if $\tilde x_{2a}>0$, i.e.\ $\tilde x_1 < 1$. I.e., intuitively speaking, type 1 can invade if and only if it can reduce the type $2a$ population size at equilibrium, and type 2 can invade if and only if it can reduce the type 1 population size at equilibrium. We do not know whether $(\tilde x_1,\tilde x_{2a})$ is stable in the coexistence regime and unstable in the founder control regime, but we expect that this is the case (cf.\ Conjecture~\ref{conj-Moran}). \color{black}

A microbial phenomenon with a similar but even stronger effect than asymmetric competition is horizontal gene transfer. There have been different ways of modelling it in the mathematical literature, see e.g.\ \cite{BCFMT16,BCFMT18,CMT19,BPT23}, but the way of modelling most similar to the instantaneous recolonization in the present paper is the one in~\cite{BT20}. Indeed, here any individual of one type is able to turn individuals of the other type into its own type, at a rate proportional to the population size of the other type. Thanks to this feature, under certain conditions even an unfit type that would be able to survive on its own can coexist with the other type (even if the latter is capable of competition-induced dormancy modelled similarly to~\cite{BT19}), which is not the case for usual asymmetric competition. Apart from the regimes of fixation and coexistence, for different choices of the parameters, a founder control regime was also identified in~\cite{BT20}.

\subsection{Interpretation of the invasion conditions} \label{sec-invasionfitnesses}

We begin with the analysis of the invasion of a single active type $2$ individual  into a resident population of type $1$. Let us recall from~\eqref{2invades1Moran} that this invasion is possible if the branching process approximating type $2$ in the initial phase is supercritical, which means that the equilibrium $(1,0)$ is unstable. This happens if and only if
\[ \kappa< \frac{\sigma}{s}\Big(\frac{1+s}{2}p-s\Big)=\kappa'.  \] 
Equivalently, we can also write this condition as 
\[s < \frac{1+s}{2} p \frac{\sigma}{\kappa+\sigma}.\numberthis\label{BT19style} \]
For $\kappa=0$, this coincides with the condition $\frac{2s}{1+s} < p$ (which is necessary for the invasion of type $2$ for any $\kappa \geq 0$ and is sufficient for $\kappa=0$\color{black}), i.e.\ the evolutionary outcome then does not depend on $\sigma$.

The right-hand side of~\eqref{BT19style} can be interpreted as the probability that in the initial phase of the invasion of type $2$, a type $2a$ individual survives a potential selective death event caused by a type $1$ individual via switching to dormancy, and then waking up before dying in the dormant state. Here, we note that when almost the entire population is of type $1$, the death of dormant individuals, due to instant recolonization by a uniformly chosen individual from the whole population, leads to a creation of type $1$ individuals with high probability. The left-hand side is the selective disadvantage of type $2a$ individuals relative to the initially resident type $1$. 
We also clearly see in~\eqref{BT19style} that if we fix all parameters but $\kappa$, then the larger $\kappa$ is, the harder it becomes for type $2$ to invade.

We can further rewrite~\eqref{BT19style} as
\[ \frac{\kappa}{\kappa+\sigma} s < \frac{\sigma}{\kappa+\sigma} \Big( \frac{1+s}{2}p -s \Big) . \]
This means that compared to the case $\kappa=0$, the selective advantage $s$ of type $1$ gets multiplied by the probability $\frac{\kappa}{\kappa+\sigma} $ that a dormant individual dies before waking up, while the benefit $\frac{1+s}{2}p -s$ of type 2 coming from its ability of selection-induced dormancy gets multiplied by the complementary probability $\frac{\sigma}{\kappa+\sigma} $, and type $2$ still has to have a larger advantage than type $1$ to invade. 
 This way, while for $\kappa=0$ the value of $\sigma$ is indifferent and only the relation between $s$ and $p$ matters, for $\sigma>0$ we see an additional trade-off between $\kappa$ and $\sigma$.

For the opposite invasion direction (type $1$ invading into type $2$ residents), recall that according to~\eqref{lambda1defMoran}, the branching process of type $1$ is supercritical, i.e.\ $(0,\bar x_{2a})$ is unstable, if and only if
\[ \kappa \bar x_{2d} > \Big( \frac{1+s}{2}p - s \Big) \bar x_{2a}. \]
Each individual in this branching process creates new individuals thanks to the death of type $2d$ individuals at rate $\kappa \bar x_{2d}$ and dies due to its disadvantage caused by the selection-induced dormancy of type $2$ at rate $ \big( \frac{1+s}{2}p - s \big) \bar x_{2a}$. The difference between this effective birth rate and death rate must be positive in order that type $1$ can invade.
It is remarkable that the invasion properties of type $1$ thus depend on the relative composition of the active and dormant sub-populations of the resident type $2$.

If we fix all parameters but $\sigma$, $\bar x_{2a}$ will tend to 1 and $\bar x_{2d}$ to 0 as $\sigma \to \infty$, and therefore for fixed $\kappa$ it becomes more and more difficult for type $1$ to invade if one increases $\sigma$.

We have seen in Theorem~\ref{theorem-mainMoran} that coexistence can only occur if and only if \[ \sigma > \frac{2s^2}{p(1-s)} \numberthis\label{sigmacond} \] and $\tilde \kappa < \kappa < \kappa'=\frac{\sigma}{s}\Big(\frac{1+s}{2}p-s\Big)$. This condition is equivalent to the simultaneous supercriticality of the type $1$ and type $2$ branching processes, i.e.\
\[ \kappa'= \frac{\sigma}{s}\Big(\frac{1+s}{2}p-s\Big) > \kappa > \frac{\bar x_{2a}}{\bar x_{2d}} \Big(\frac{1+s}{2}p-s\Big). \]
Thus, for the existence of such a $\kappa$ it is necessary that
$\sigma \bar x_{2d} > s \bar x_{2a}.$
Using that $\bar x_{2d} = 1 - \bar x_{2a}$, the latter condition is equivalent to 
\[ \bar x_{2a} < \frac{\sigma}{s+\sigma}. \]
This gives an upper bound on the size of the active equilibrium population size of type $2$ such that mutual invasion, and in particular of type $1$, can be possible.

\begin{remark}\label{remark-sigma}
As we have seen, condition~\eqref{sigmacond} is equivalent to the existence of $\kappa$ for which there is coexistence, while $\sigma < \frac{2s^2}{p(1-s)}$ is equivalent to the presence of founder control for some appropriately chosen $\kappa$. We often write these inequalities formally as conditions on $\sigma$ because this is convenient for some proofs, but note that the roles of $p$ and $\sigma$ in these inequalities are actually symmetric. From a biological point of view, we think that the most suitable way of writing~\eqref{sigmacond} is the following:
\[ p\sigma > \frac{2s^2}{1-s}. \]
This means on the one hand that even if the selection-induced switching mechanism is not so strong (i.e.\ $p$ is relatively small), a fast resuscitation (i.e.\ large $\sigma$) can compensate this in order to reach coexistence (for some $\kappa$). On the other hand, a slow resuscitation (i.e.\ small $\sigma$) can (for some $\kappa$) be compensated by an efficient selection-induced switching to dormancy (i.e.\ large $p$).
\end{remark}
\normalcolor

\subsection{Stable vs metastable coexistence}
\label{sec-metastability}
An ubiquitous and well-known caveat that arises when stochastic individual based systems are approximated by ODEs is that stable equilibria in the dynamical system typically do not imply perpetual \igb stability \bl in the stochastic system, \igb due to random fluctuations, \bl even in the Moran model set-up, where the overall population size remains constant for all times.

In particular, \igb in our set-up, \bl even if one can verify that under the coexistence scenario (iii) the system~\eqref{2dMoran} converges to $(\tilde x_1, \tilde x_{2a})$ as time tends to infinity, this does not mean simultaneous perpetual non-extinction of types 1 and 2 in the Moran model. Indeed, the Markov chain $(\mathbf X^{(N)}(t))_{t\geq 0}$ \igb {\em will} \bl almost surely either reach the absorbing state  $(N,0,0)$ or the subset $\{ (0,x_{2a},N-x_{2a}) \colon x_{2a} \in \{ 0,1,\ldots,N\} \}$ of the state space and stay there forever (restricted to the latter subset, the Markov chain is irreducible, and hence almost surely neither type $2a$ nor type $2d$ can go extinct forever). Thus, either type $1$ or type $2$ will die out in almost surely finite time for fixed $N$, leading to the fixation of the other type. Let us note that in models with logistic competition, not only coexistence is \igb merely \bl metastable, but \igb also \bl the entire population dies out in almost surely finite time \igb (in contrast to the Moran model set-up). However, \bl the expected extinction time grows exponentially with the carrying capacity.


\section{\igb Preparatory analysis of \bl the dynamical system~\eqref{2dMoran}}
\label{sec-preliminaryproofsMoran}
The goal of the present section is to prove preliminary results on the dynamical system~\eqref{2dMoran} that we will need for the proof of Theorem~\ref{theorem-mainMoran}. First, in Section~\ref{sec-stability}, we prove results on the existence and stability of equilibria of the dynamical system that were already stated in Section~\ref{sec-modeldefmainresultsMoran}. Based on these results and their proofs, in Section~\ref{sec-convergenceMoran} we state and prove lemmas on the asymptotic behaviour of the system. 

\subsection{Existence and stability of equilibria}\label{sec-stability}
In this section we \igb provide \bl  our previously stated assertions on the existence and stability of equilibria of the system~\eqref{2dMoran}: Lemmas~\ref{lemma-simultaneousstability} and~\ref{lemma-xtildeexistence}.

\begin{proof}[Proof of Lemma~\ref{lemma-simultaneousstability}]

Let us first analyse the stability of \igb the monomorphic equilibrium \bl $(1,0)$. The Jacobi matrix at this \igb point \bl is
\[ 
J_{(1,0)}:=\begin{pmatrix}
    -\kappa & s-\frac{1+s}{2} p -\kappa\\
     -\sigma & -\sigma -s
\end{pmatrix}. 
\]
Its trace is negative and \igb its determinant satisfies \bl $\det J_{(1,0)}= \sigma(s-\frac{1+s}{2}p)+\kappa s$, and so \igb from the basic linearization theorem \bl we can derive the following conclusions:
\begin{align*}
    & (a) \quad \text{stability of } (1,0) \Leftrightarrow \kappa > \kappa'= \frac{\sigma}{s}\Big( p\frac{1+s}{2}-s\Big) , \\
   &  (b) \quad \text{instability of } (1,0) \Leftrightarrow 0 < \kappa < \kappa'=\frac{\sigma}{s}\Big( p\frac{1+s}{2}-s \Big).
\end{align*}
\igb Regarding the \bl stability of \igb the other monomorphic equilibrium \bl $(0,\bar x_{2a})$, \igb we obtain for the corresponding \bl Jacobi matrix:
\[ 
J_{(0,\bar{x}_{2a})}:=\begin{pmatrix}
    \bar{x}_{2a}(s-\frac{1+s}{2}p) +\kappa(1-\bar{x}_{2a}) & 0\\
     -\bar{x}_{2a} s -\kappa\bar{x}_{2a} - \sigma & -2p\bar{x}_{2a}-\kappa\bar{x}_{2a}-\sigma+\kappa(1-\bar{x}_{2a})
\end{pmatrix}. 
\]
Let us first observe that the bottom right entry is always negative, indeed,
\[\begin{aligned}
    -2p\bar{x}_{2a}-\kappa\bar{x}_{2a}-\sigma+\kappa(1-\bar{x}_{2a}) <0 &\Leftrightarrow \bar{x}_{2a} > \frac{\kappa - \sigma}{2(p+\kappa)},  
\end{aligned}
\]
which is always \igb satisfied thanks to \bl the expression~\eqref{x2abar,kappa>0} \igb for \bl  $\bar{x}_{2a}$. Then, the stability only depends on the sign of the determinant, and so on the sign of the upper left entry. 

When $s-\frac{1+s}{2}p-\kappa>0$, stability is equivalent to
\[ \bar{x}_{2a} < \frac{\kappa}{\kappa -s +\frac{1+s}{2}p}, \numberthis\label{x2aupper} \] 
which is impossible as the right-hand side is negative, and $\bar{x}_{2a}$ is positive. We then have instability when $s-\frac{1+s}{2}p>\kappa$. \igb For \bl $\kappa>s-\frac{1+s}{2}p$, we have stability when
\[\bar{x}_{2a} > \frac{\kappa}{\kappa -s +\frac{1+s}{2}p}, \numberthis\label{x2alower} \] 
which is impossible \igb for \bl  $p<\frac{2s}{1+s}$, because $\bar{x}_{2a}<1$, we have instability in that case. And when  $p>\frac{2s}{1+s}$, using~\eqref{x2abar,kappa>0} again, we must have
\[\frac{\kappa - \sigma + \sqrt{(\kappa-\sigma)^2+4\sigma(p+\kappa)}}{2(p+\kappa)} > \frac{\kappa}{\kappa -s +\frac{1+s}{2}p}.\] 
Writing $R=\kappa -s +\frac{1+s}{2}p$ ($R>0$ here), we must have
\begin{align*}
    &R\sqrt{(\kappa-\sigma)^2+4\sigma(p+\kappa)}> 2\kappa(p+\kappa) - R(\kappa-\sigma) \\
    &\Leftrightarrow R^2((\kappa-\sigma)^2+4\sigma(p+\kappa)) > 4(p+\kappa)^2\kappa^2 + R^2(\kappa-\sigma)^2 - 4\kappa (p+\kappa)(\kappa-\sigma)R \\
    &\Leftrightarrow \tilde{Q}(\kappa):= (\kappa+p)(\kappa^2(-4p+2D)+\kappa 2D\sigma +D^2\sigma) > 0, \numberthis\label{Qfunction}
\end{align*}
\igb with \bl 
\[ D=p(s+1)-2s \numberthis\label{Ddef} \]
and where for the first equivalence, we used that the right-hand side of the first line is positive for $\kappa>0$. Let us also remark that $D-2p = p(s-1)-2s$. Considering the second degree polynomial $\tilde{Q}(\kappa)$ that we have in $\kappa$, the leading coefficient of $\tilde{Q}$ is negative, $\tilde{Q}(0)>0$, and so there is a unique root in $(0,\infty)$. Knowing that $D-2p<0$, this unique root must be
\[\tilde{\kappa}=\frac{-2D\sigma-\sqrt{4D^2\sigma^2-4(-4p+2D)\sigma D^2}}{4D-8p} \]
(cf.\ \eqref{tildekappadef}). 
And so, we obtain the instability of $(0,\bar{x}_{2a})$ when $\kappa>\tilde{\kappa}$ and the stability of $(0,\bar{x}_{2a})$ when $0<\kappa<\tilde{\kappa}$. Summarizing, to have stability, the following three conditions must be satisfied:
\begin{align*}
    &(0) \quad \kappa > s-\frac{1+s}{2}p, \\
    &(1) \quad p> \frac{2s}{1+s}, \\
    &(2) \quad 0<\kappa<\tilde{\kappa}.
\end{align*}
Condition $(0)$ can be ignored \igb since it is automatically implied by $(1)$ combined with \bl
$\kappa>0$. 

If $p<\frac{2s}{1+s}$, then $(0,\bar{x}_{2a})$ is unstable because of condition $(1)$ and $(1,0)$ is stable, because we have $\kappa>\frac{\sigma}{s}\Big(\frac{1+s}{2}p-s\Big)=\kappa'$.

The instability of $(1,0)$ corresponds to the case where $0<\kappa<\kappa'$. In that case, $p>\frac{2s}{1+s}$. Further we have that $\tilde{\kappa}<\kappa'$ when 
\[\begin{aligned}
    &~\frac{\sigma}{s} \Big( \frac{1+s}{2}p-s \Big) > \frac{-2D\sigma-\sqrt{4D^2\sigma^2-4(-4p+2D)\sigma D^2}}{4D-8p} \\
    &\Leftrightarrow 4(D-2p)\frac{D}{2}\frac{\sigma}{s}+2D\sigma < -\sqrt{4D^2\sigma^2-4(-4p+2D)\sigma D^2} \qquad \text{as } D-2p<0 \\
    &\Leftrightarrow 2\sigma D \Big( \frac{D-2p}{s} +1 \Big)  < -\sqrt{4D^2\sigma^2-4(-4p+2D)\sigma D^2} \\
    &\Leftrightarrow 4\sigma^2D^2 \Big( \frac{D-2p}{s}+1 \Big)^2 > 4D^2\sigma^2 - 4(-4p+2D)\sigma D^2 \qquad \text{as } \frac{D-2p}{s} +1<0 \Leftrightarrow p(s-1)<s \\
    &\Leftrightarrow 4\sigma^2D^2 \Big( \frac{D-2p}{s} \Big)^2 + 8\sigma^2D^2\frac{D-2p}{s} > -4(-4p+2D)\sigma D^2 \\
    &\Leftrightarrow \sigma^2\Big( \frac{D-2p}{s} \Big)^2+2\sigma^2\frac{D-2p}{s} > -(2D-4p)\sigma \\
    &\Leftrightarrow \frac{\sigma}{2s^2}(D-2p) + \frac{\sigma}{s} < -1 \\
    &\Leftrightarrow \frac{D-2p}{2s} +1 < -\frac{s}{\sigma} \\
    &\Leftrightarrow \sigma > \frac{2s^2}{p(1-s)}, 
\end{aligned}\]

Then, the condition to have $\tilde{\kappa}<\kappa'$ is:
\[\sigma > \frac{2s^2}{p(1-s)}. \numberthis\label{coexistencecond}\]

Now, if $p>\frac{2s}{1+s}$, we have:

\begin{align*}
    &(S_1) \quad \text{When } \kappa \in (0,\tilde{\kappa}), \: (0,\bar{x}_{2a}) \text{ is stable}, \\
    &(S_2) \quad \text{When } \kappa \in (0, \kappa'), \: (1,0) \text{ is unstable, } \\
    &(S_3) \quad \text{When } \kappa \in (\kappa', \infty), \: (1,0) \text{ is stable, } \\
    &(S_4) \quad \text{When } \kappa \in (\tilde{\kappa}, \infty), \: (0,\bar{x}_{2a}) \text{ is stable}.
\end{align*}

Further, if $\sigma > \frac{2s^2}{p(1-s)}$, then $\tilde{\kappa}<\kappa'$, and we have that both equilibria are unstable in $(\tilde{\kappa},\kappa')$, and otherwise both equilibria are stable in the interval with reverse endpoints. This concludes the proof. 
\end{proof}

\begin{proof}[Proof of Lemma~\ref{lemma-xtildeexistence}]

\igb Observe \bl 
that in the general case, the Jacobi matrix at an \igb (potentially polymorphic) \bl equilibrium $(\tilde{x}_1,\tilde{x}_{2a})$ is \igb given by \bl
\[ 
J_{(\tilde{x}_1,\tilde{x}_{2a})}:=\begin{pmatrix}
    \tilde{x}_{2a}(s-\frac{1+s}{2}p) +\kappa(1-\tilde{x}_1 - \tilde{x}_{2a}) -\kappa \tilde{x}_1 & \tilde{x}_1(s-\frac{1+s}{2}p) - \kappa \tilde{x}_1\\
     -\tilde{x}_{2a} s -\kappa\tilde{x}_{2a} - \sigma & -\tilde{x}_1 s-2p\tilde{x}_{2a}-\kappa\tilde{x}_{2a}-\sigma+\kappa(1-\tilde{x}_1 - \tilde{x}_{2a})
\end{pmatrix}. 
\]

\igb In particular, if \bl 
we assume that $(\widetilde x_1,\widetilde x_{2a})$ is \igb indeed \bl coordinate-wise positive, we have that $\tilde{x}_{2a}(s-\frac{1+s}{2}p) +\kappa(1-\tilde{x}_1 - \tilde{x}_{2a})=0$ in order to have $\dot{x}_1=0$, and so the Jacobi matrix \igb reduces to \bl
\[ 
J_{(\tilde{x}_1,\tilde{x}_{2a})}:=\begin{pmatrix}
     -\kappa \tilde{x}_1 & \tilde{x}_1(s-\frac{1+s}{2}p) - \kappa \tilde{x}_1\\
     -\tilde{x}_{2a} s -\kappa\tilde{x}_{2a} - \sigma & -\tilde{x}_1 s-2p\tilde{x}_{2a}-\kappa\tilde{x}_{2a}-\sigma+\kappa(1-\tilde{x}_1 - \tilde{x}_{2a})
\end{pmatrix}. 
\]

Let us \igb now investigate explicit expressions for such a truly polymorphic equilibrium. It \bl
must satisfy the equations
\[
\begin{aligned}
      0& = x_{1}(t) \Big( x_{2a}(t) (s-\frac{1+s}{2}p)+\kappa (1-x_1(t)-x_{2a}(t)) \Big), \\
    0 &= -x_{1}(t)x_{2a}(t) s - x_{2a}(t)^2 p + (1-x_{1}(t)-x_{2a}(t))(\kappa x_{2a}(t)+\sigma).
\end{aligned} 
\]
From the first equation, we obtain
\[\tilde{x}_{2a}=\kappa(\tilde{x}_1-1)\frac{1}{s-\frac{1+s}{2}p-\kappa},\]
and from the second:
\[\tilde{x}_1=\frac{-\tilde{x}_{2a}^2p+(1-\tilde{x}_{2a})(\kappa\tilde{x}_{2a}+\sigma)}{\tilde{x}_{2a}+\kappa\tilde{x}_{2a}+\sigma},\]
and putting the equations together \igb yields \bl 
\[\begin{aligned}
    \tilde{x}_{2a}&=\kappa\Big(\frac{-\tilde{x}_{2a}^2p+(1-\tilde{x}_{2a})(\kappa\tilde{x}_{2a}+\sigma)}{\tilde{x}_{2a}+\kappa\tilde{x}_{2a}+\sigma}-1\Big)\frac{1}{s-\frac{1+s}{2}p-\kappa} \\
    &=\kappa\frac{-\tilde{x}_{2a}^2p-\tilde{x}_{2a}(\kappa \tilde{x}_{2a}+\sigma) -\tilde{x}_{2a}s}{\tilde{x}_{2a}s+\kappa\tilde{x}_{2a}+\sigma}\frac{1}{s-\frac{1+s}{2}p-\kappa}.
\end{aligned}\]
After dividing on both sides by $\tilde{x}_{2a}$ and rearranging the terms, we get
\[\tilde{x}_{2a}=\frac{2\kappa s+\sigma(2s-p(1+s))}{\kappa p(s-1)+s(-2s+p(s+1))}. \numberthis\label{tildex2a} \]
We can already see \igb from this \bl that $\lim_{\kappa \rightarrow 0}\tilde{x}_{2a}=-\sigma/s $, and so the coordinate-wise positive equilibrium does not exist for $\kappa$ small enough. Further, $\tilde{x}_{2a}$ is negative if $p<\frac{2s}{1+s}$, \igb so that we now \bl suppose that $p>\frac{2s}{1+s}$. We also always have
\igb
that
\begin{equation}
    \label{eq_x1}
    \tilde{x}_1= \tilde{x}_{2a} \Big( s-\frac{1+s}{2}p-\kappa \Big) \frac{1}{\kappa}+1,
\end{equation}
and 
\begin{equation*}
    \label{eq_x2a}
    \tilde{x}_{2d}=\frac{\tilde{x}_{2a}}{\kappa}\Big( \frac{1+s}{2}p-s \Big). 
\end{equation*}
\bl
In order \igb for all \bl those three coordinates \igb to \bl belong to $(0,1)$, we need that
\[0< \tilde{x}_{2a} < \frac{\kappa}{\kappa+\frac{1+s}{2}p-s}. \]

\igb Of particular interest is the existence of  the \bl 
coordinate-wise positive equilibrium \igb in the case when the two monomorphic equilibria \bl are unstable (which is possible only if $\sigma > \frac{2s^2}{p(1-s)}$). First, we \igb note \bl  that when $\kappa=\kappa'=\frac{\sigma}{s}(\frac{1+s}{2}p-s)$, then the numerator of $\tilde{x}_{2a}$ is equal to 0, and so is $\tilde{x}_{2a}$, and then $\tilde{x}_1=1$ because of 
\igb \eqref{eq_x1} \bl
above. Also, when $\kappa=\tilde{\kappa}$, we have
\[\begin{aligned}
    \tilde{x}_1=0 &\Leftrightarrow (2\tilde{\kappa}s-\sigma D)(-2\tilde{\kappa}-D)=(\tilde{\kappa}p(1-s)-sD)2\tilde{\kappa} \\
    & \Leftrightarrow \tilde{\kappa}^2(-4s-2p(1-s))+2\tilde{\kappa}\sigma D+\sigma D^2=0
\end{aligned}\]
(recalling the definition of $D$ from~\eqref{Ddef}), which is indeed the case for $\tilde{\kappa}$ by computing the roots of the second degree polynomial $\tilde Q$ defined in ~\eqref{Qfunction}. And so, as $\tilde{x}_1=0$ at $\tilde{\kappa}$, we have
\[\tilde{x}_{2a}=\frac{\tilde{\kappa}}{\tilde{\kappa} +\frac{1+s}{2}p-s},\]
at $\tilde{\kappa}$, and this is precisely the value of $\bar{x}_{2a}$ at $\tilde\kappa$ as we have seen in~\eqref{x2aupper}--\eqref{x2alower}.

\medskip

Now we assume that $\sigma > \frac{2s^2}{p(1-s)}$ (we will comment on the case  $\sigma < \frac{2s^2}{p(1-s)}$ at the end of the proof). To finish the proof of Lemma~\ref{lemma-xtildeexistence}, we need to show that the coordinate-wise positive equilibrium $(\tilde x_{1}, \tilde x_{2a})$ exists on $(\tilde{\kappa},\kappa')$ (i.e.\ when the two monomorphic equilibria are unstable).

In the first boundary point $\tilde{\kappa}$, this equilibrium will be equal to $(0,\bar{x}_{2a})$, which is the stable equilibrium before the interval we consider, and on the second boundary point $\kappa'$, this equilibrium will be equal to $(1,0)$, the stable equilibrium after the interval we consider. For the existence of the coordinate-wise positive equilibrium on the interval, we need in particular that the denominator of $\tilde{x}_{2a}$ (cf.\ \eqref{tildex2a}) does not vanish on this interval. Moreover, as $\lim_{\kappa\rightarrow 0}\tilde{x}_{2a}<0$, $\tilde{x}_{2a}=\bar{x}_{2a}$ in $\tilde{\kappa}$ and $\tilde{x}_{2a}$ only vanishes in $\kappa'$ which is greater than $\tilde{\kappa}$, we would like to show that the asymptote (where the denominator vanishes) appears in $(0,\tilde{\kappa})$ (See the right image of Figure~\ref{fig-equilibria} for an illustration).

\medskip

Using the expression of $\tilde{x}_{2a}$, we have the asymptote when
\[\kappa=\frac{s(-2s+p(s+1))}{p(1-s)}. \]
Also, the denominator vanishes before $\kappa'=\frac{\sigma}{s}\Big(p\frac{1+s}{2}-s\Big)$ if
\[\begin{aligned}
    &\quad \frac{s(-2s+p(s+1))}{p(1-s)}<\frac{\sigma}{s}\Big(p\frac{1+s}{2}-s\Big) \\
    &\Leftrightarrow \frac{s}{p(1-s)}< \frac{\sigma}{2s} \\
    &\Leftrightarrow \sigma > \frac{2s^2}{p(1-s)},
\end{aligned}\]
which is precisely the condition we have in ~\eqref{coexistencecond}, which guarantees that we have an interval where the previous two equilibria, $(1,0)$ and $(0,\bar x_{2a})$, are unstable. 

\bigskip

Moreover, when $\sigma>\frac{2s^2}{p(1-s)}$, we obtain, by considering $\tilde{\kappa}$ as a function of $\sigma$, that
\[
    \tilde{\kappa}>\frac{\frac{4Ds^2}{p(1-s)}+\frac{4Ds}{p(1-s)}\sqrt{s^2+(2p-D)p(1-s)}}{4(2p-D)}. 
\]
Also:
\[\begin{aligned}
    &\quad \frac{Ds +D\sqrt{s^2+(2p-D)p(1-s)}}{2p-D} = D \\
    &\Leftrightarrow s + \sqrt{s^2+(2p-D)p(1-s)}=2p-D \\
    &\Leftrightarrow s^2+(2p-D)p(1-s) =(p(1-s)+s)^2 \\
    &\Leftrightarrow (2p-D)p(1-s) = p^2(1-s)^2 +2ps(1-s) \\
    &\Leftrightarrow 2p-D=p(1-s)+2s,
\end{aligned}\]
which is always the case as $D=p(s+1)-2s$. And using that the asymptote occurs when $\kappa=\frac{sD}{p(1-s)}$, we have that when $\sigma>\frac{2s^2}{p(1-s)}$, then $\tilde{\kappa}>\frac{sD}{p(1-s)}$. To summarize, when we have an interval where the first two equilibria are unstable, then the asymptote of $\tilde{x}_{2a}$ occurs before $\kappa=\kappa'=\frac{\sigma}{s}\big(p\frac{1+s}{2}-s\big)$, and also before $\kappa=\tilde{\kappa}$. In particular, it explains why $\tilde{x}_{2a}$ goes from $-\sigma/s$ to $\bar{x}_{2a}$ on $(0,\tilde{\kappa})$ and does not vanish on this interval, as it only vanishes at $\kappa'>\tilde{\kappa}$.

\medskip

Writing $\tilde{x}_{2a}=f(\kappa)$, we can rewrite $\tilde{x}_1$ and $\tilde{x}_{2d}$ as
\[\begin{aligned}
    &\tilde{x}_1=f(\kappa)\Big( s-\frac{1+s}{2}p-\kappa \Big) \frac{1}{\kappa}+1 \\
    &\tilde{x}_{2d}=\frac{f(\kappa)}{\kappa}\Big( \frac{1+s}{2}p-s \Big).
\end{aligned}\]
For such a function $f$, we have
\[\begin{aligned}
    f'(\kappa)&=\frac{2s(\kappa p(s-1)+s(-2s+p(s+1)))-p(s-1)(2\kappa s+\sigma(2s-p(1+s)))}{(\kappa p(s-1)+s(-2s+p(s+1)))^2} \\
    &= \frac{2s^2(-2s+p(s+1))-\sigma p(s-1)(2s-p(1+s))}{(\kappa p(s-1)+s(-2s+p(s+1)))^2}.
\end{aligned}\]

And so $f'(\kappa)<0 \Leftrightarrow \sigma > \frac{2s^2}{p(1-s)}$, which is the case here. From this, we obtain that $\tilde{x}_{2a}$ and $\tilde{x}_{2d}$ are decreasing, and so $\tilde{x}_1=1-\tilde{x}_{2a}-\tilde{x}_{2d}$ is increasing. Using the results we have above about the values of $(\tilde{x}_1,\tilde{x}_{2a})$ at $\tilde{\kappa}$ and $\kappa'$, we conclude that the coordinate-wise positive equilibrium exists precisely on the interval $(\tilde{\kappa},\kappa')$ when $\sigma>\frac{2s^2}{p(1-s)}$. 

Also, when $\sigma<\frac{2s^2}{p(1-s)}$, this equilibrium exists if and only if $\kappa\in( \kappa',\tilde\kappa)$. Indeed, we obtain, similarly as above, that the asymptote is not only greater than $\kappa'$ but also larger than $\tilde\kappa$ (See the left image of Figure~\ref{fig-equilibria} for a visualization). And we obtain all the results formulated in that case using the proofs of the ones formulated for $\sigma>\frac{2s^2}{p(1-s)}$.
\end{proof}

\subsection{Convergence of the dynamical system}\label{sec-convergenceMoran}

\begin{lemma}\label{lemma-convtoaneq}
Consider the dynamical system~\eqref{2dMoran}. Assume that $x_{1}(0),x_{2a}(0)>0$, $x_{1}(0)+x_{2a}(0)<1$, and that the assumptions of Theorem~\ref{theorem-mainMoran} on $p,s,\sigma,\kappa$ hold. Then, $(x_{1}(t),x_{2a}(t))$ has no periodic orbit lying in $\bar{\Delta}$, where $\Delta$ is the open and simply connected subset
\[ \Delta = \{ (x_1,x_{2a}) \in \R^2 \colon x_1,x_{2a}>0, x_1+x_{2a}<1 \} \subset \R^2, \]
which is positively invariant under the flow of~\eqref{2dMoran}.
\end{lemma}
\begin{proof}
The divergence of system~\eqref{2dMoran} is 
\[  x_{2a}(t) \big( s-\frac{1+s}{2}p \big) - x_{1}(t)s - 2 x_{2a}(t) p -\sigma+\kappa(2-3 x_{1}(t)-3 x_{2a}(t)). \numberthis\label{divergence} \]
 Writing, for any $\kappa \geq 0$,
\[
\begin{aligned}
      f(x_1,x_{2a}) &= x_{1} \Big( x_{2a} \big(s-\frac{1+s}{2}p\big)+\kappa x_{2d} \Big), \\
     g(x_1,x_{2a}) &= -x_{1}x_{2a} s - x_{2a}^2 p + x_{2d}(\kappa x_{2a}+\sigma),
\end{aligned} 
\]
where $x_{2d}=x_{2d}(x_1,x_{2a})=1-x_1-x_{2a}$, 
we have $\dot{x}_1(t)=f(x_1(t),x_{2a}(t))$, $\dot{x}_{2a}(t)=g(x_1(t),x_{2a}(t))$ and thus
\[
\begin{aligned}
      \partial_{x_1} \Big(\frac{f(x_1,x_{2a})}{x_1x_{2a}}\Big) &= \partial_{x_1}\Big( s-\frac{1+s}{2}p+\kappa \frac{x_{2d}}{x_{2a}} \Big)=-\frac{\kappa}{x_{2a}}, \\
     \partial_{x_{2a}} \Big(\frac{g(x_1,x_{2a})}{x_1x_{2a}}\Big) &= \partial_{x_{2a}}\Big(-s - \frac{x_{2a}}{x_1} p + \frac{x_{2d}}{x_1}\kappa +\sigma\frac{x_{2d}}{x_1x_{2a}}\Big)=-\frac{p}{x_1}-\frac{\kappa}{x_1}-\sigma\frac{x_{2d}}{x_1x_{2a}^2}-\frac{\sigma}{x_{2a}x_1},
\end{aligned} 
\]
and thanks to the Bendixson--Dulac criterion, using that the divergence we obtained is negative in the open and simply connected subset $\Delta$, where the multiplicator function $(x_1,x_{2a}) \mapsto \frac1{x_1x_{2a}}$ is continuously differentiable, it follows that the system has no periodic orbit lying entirely in $\Delta$. 

By the assumption of the lemma, the initial condition $x_{1}(0)+x_{2a}(0)$ lies in $\Delta$, and the closure $\overline{\Delta}$ of $\Delta$ is positively invariant under the system~\eqref{2dMoran}. Now we show that there is no nontrivial periodic orbit in $\overline{\Delta}$ that has a point in $\Delta$. Note that if such a periodic orbit exists, then by a timeshift we also obtain such a periodic orbit starting in $\Delta$ at time $0$. So assume for a contradiction that there exists such a periodic orbit $\varphi=(\varphi_{t})_{t\geq 0}=(\varphi^{(1)}_{t},\varphi^{(2a)}_{t})_{t\geq 0}$ with period $T>0$ and with $\varphi_0 \in \Delta$. Then this trajectory cannot lie entirely in $\Delta$, so it must hit $\partial\Delta$ at some finite time. Let $t=\inf \{ s \in [0,T] \colon \varphi_s \in \partial \Delta \}$. Now we show that $t<\infty$ leads to a contradiction in all possible cases.

{\bf Case 1}: $\varphi_t=(x_1,0)$ for some $0 \leq x_1 < 1 $. Then, the right-hand side of the second equation of~\eqref{2dMoran} is $\sigma(1-x_1)>0$. It follows by continuity that there exists some $\eps>0$ such that $\dot \varphi^{(2a)}_u>0$ for all $u \in (t-\eps,t+\eps)$. However, this contradicts the assumption that $\varphi_\cdot$ hits $\{ (y,0) \colon y \in [0,1] \} \subseteq \partial \Delta$ at time $t$.

{\bf Case 2}: $\varphi_t=(1,0)$. Then, since $(1,0)$ is an equilibrium of~\eqref{2dMoran}, it follows that $\varphi_s=(1,0)$ for all $s \geq t$. This contradicts the assumption that $\varphi$ is a nontrivial periodic orbit.

{\bf Case 3}: $\varphi_t=(0,x_{2a})$ for some $0<x_{2a} \leq 1$. Then, at time $t$, $\dot \varphi^{(1)}$ vanishes and stays zero forever. Now, for $s>t$, we have
\[ \dot \varphi^{(2a)}_s= 
(1-\varphi_s^{(2a)})(\kappa \varphi_s^{(2a)} + \sigma),
\]
so $(\varphi^{(2a)}_s)_{s \geq t}$ is the solution of a one-dimensional autonomous ODE, which cannot have a nontrivial periodic orbit (in fact, it clearly converges to the stable equilibrium $\bar x_{2a}$ of the one-dimensional system as time tends to infinity). 

{\bf Case 4}: $\varphi_t^{(1)}, \varphi_t^{(2a)} >0$ and $\varphi_t^{(1)}+\varphi_t^{(2a)} =1$. This case is similar to Case 1. Indeed, now for $\varphi_t^{(2d)}(\cdot):=1-\varphi_t^{(1)}(\cdot)-\varphi_t^{(2a)}(\cdot)$ we obtain (cf.\ \eqref{2dexpress})
\[ 
\dot \varphi^{(2d)}_t = \varphi^{(1)}_t \varphi^{(2a)}_t \frac{1+s}{2}p + (\varphi^{(2a)}_s)^2 p >0.
\]
Thus, by continuity, there must exist some $\eps>0$ such that for all $u \in (t-\eps,t+\eps)$, $\dot \varphi^{(2d)}_u > 0$. But this  contradicts the assumption that $\varphi_\cdot$ hits $\{ (x_1,x_{2a}) \colon x_1,x_{2a}>0, x_1+x_{2a}=1 \} \subseteq \partial \Delta$ at time $t$.

We conclude that if there exists a nontrivial periodic orbit in $\overline{\Delta}$, it must lie entirely in $\partial \Delta$. However, our argument in Cases 1 and 2 shows that such a trajectory cannot pass through a point of the form $(x_1,0)$ with $0 \leq x_1 \leq 1$: For $x_1<1$ the problem is that the trajectory must enter $\Delta$, and for $x_1=1$ it degenerates to an equilibrium. Moreover, our reasoning in Case 3 shows that such a periodic trajectory cannot pass through a point of the form $(0,x_{2a})$ with $0<x_{2a} \leq 1$ because it would converge to $(0,\bar x_{2a})$ as time tends to infinity. Finally, by Case 4, it cannot pass through a point of the form $(x_1,x_{2a})$ with $x_1,x_{2a}>0$, $x_1+x_{2a}=1$ either because it would also enter $\Delta$.

It follows that system~\eqref{2dMoran} has no periodic orbit in $\overline{\Delta}$. 
\end{proof}


\begin{remark}
Since $x_1(t)+x_{2a}(t) < 1$, the divergence~\eqref{divergence} is negative on the entire set $\Delta$ whenever $\kappa < \sigma/2$ (in particular if $\kappa=0$). In this case, it follows immediately from the Bendixson criterion that the system has no periodic orbit lying entirely in $\Delta$, which simplifies the proof of Lemma~\ref{lemma-convtoaneq}.
\end{remark}
\color{black}
\begin{lemma}\label{lemma-convtoaneq2}
Consider the dynamical system~\eqref{2dMoran} under the assumption $p>\frac{2s}{1+s}$ and for $\kappa>0$. Under the condition $x_1(0)+x_{2a}(0)<1$ and supposing that the system does not start at an equilibrium, we have the following convergences for $\sigma>\frac{2s^2}{p(1-s)}$:
\begin{enumerate}[1)]
\item $\text{ If } \kappa\in(0,\tilde{\kappa}) \text{, the system converges to } (0,\bar{x}_{2a})$.
\item $ \text{ If } \kappa\in(\tilde{\kappa},\kappa') \text{, the system converges along a subsequence to } (\tilde{x}_1,\tilde{x}_{2a})$.
\item $\text{ If } \kappa \in(\kappa',\infty) \text{, the system converges to } (1,0).$
\end{enumerate}
And when $\sigma<\frac{2s^2}{p(1-s)}$, we have $\kappa'<\tilde{\kappa}$ and  the following convergences:
\begin{enumerate}[1)]
    \item $\text{ If } \kappa\in(0,\kappa') \text{, the system converges to } (0,\bar{x}_{2a})$.
    \item $\text{ If } \kappa \in(\tilde{\kappa},\infty) \text{, the system converges to } (1,0).$
\end{enumerate}
\end{lemma}

\begin{remark}
\igb We refrain from investigating the convergence for \bl 
$\kappa \in (\kappa',\tilde{\kappa})$ when $\sigma<\frac{2s^2}{p(1-s)}$.  {\igb Here,} what we know is that both $(1,0)$ and $(0,\bar x_{2a})$ are locally asymptotically stable, so they \igb each \bl  have a domain of attraction with a nonempty interior. \color{black} \igb However, this case is not relevant \bl 
for \igb our study, \bl as the approximating branching process \igb for \bl the mutant population in our stochastic individual-based model during the first phase of invasion is subcritical (thanks to Lemma~\ref{lemma-phase1Moran} below), and thus a mutant invasion is unsuccessful with high probability, so that a dynamical system approximation is \igb meaningless. \bl
\end{remark}
\begin{proof}[Proof of Lemma~\ref{lemma-convtoaneq2}]

Let us suppose that the solution converges to $(1,0)$ as $t\to\infty$ when $\kappa<\kappa'$. Then writing 
\[ g_{2a}(x_1,x_{2a}) = -x_1 x_{2a} s - x_{2a}^2 p + (1-x_1-x_{2a})(\kappa x_{2a}+\sigma) \]
and noting that $\dot x_{2a}(t)=g_{2a}(x_{1}(t),x_{2a}(t))$ according to the second  equation of~\eqref{2dMoran}, we have thanks to the continuity of $g_{2a}$ that
\[ \lim_{t\to\infty} \dot x_{2a}(t) = \lim_{t\to\infty} g_{2a}(x_1(t),x_{2a}(t)) = g_{2a}(1,0)=0. \]
Thus, for $\eps>0$ fixed we can find $t_0=t_0(\eps) >0$ such that for all $t \geq t_0$ we have
\[x_{2d}(t) < \frac{x_1(t)x_{2a}(t)s+x_{2a}^2(t)p}{\kappa x_{2a}(t)+\sigma} + \eps,\]
and then
\[\dot x_1(t) < x_1(t)\Big( x_{2a}(t)(s-\frac{1+s}{2}p)+\kappa\big(\frac{x_1(t)x_{2a}(t)s+x_{2a}^2(t)p}{\kappa x_{2a}(t)+\sigma} + \eps\big) \Big).\]
And also
\[\dot x_1(t)<x_1(t)x_{2a}(t)\Big( s-\frac{1+s}{2}p+x_1(t)\kappa \frac{s}{\sigma}+x_{2a}(t)\kappa\frac{p}{\sigma} \Big) + \kappa \eps.\]
Using that $\lim_{t\to\infty}x_1(t)=1$, $\lim_{t\to\infty}x_{2a}(t)=0$ and $\kappa\frac{s}{\sigma}<p\frac{1+s}{2}-s$, when $\eps$ is small enough we have for $t>t_0$ large enough that
\[\dot x_1(t)<0,\]
which is a contradiction, as $x_1(t)$ must tend to 1. Thus, the solution cannot converge to $(1,0)$ when $\kappa<\kappa'$.
\color{black}
\bigskip

Let us suppose that the solution converges to $(0,\bar{x}_{2a})$ when $\kappa>\tilde{\kappa}$. Using $\kappa>\tilde{\kappa}$, we have that $\tilde{Q}(\kappa)<0$ ($\tilde{Q}$ defined in ~\eqref{Qfunction}), and so:
\[
\begin{aligned}
    &\qquad \bar{x}_{2a} < \frac{\kappa}{\kappa+\frac{1+s}{2}p-s} \\
    &\Leftrightarrow \bar{x}_{2a}\Big( s-\frac{1+s}{2}p \Big) + \kappa (1-\bar{x}_{2a})>0.
\end{aligned}\]
Using that $\lim_{t\rightarrow\infty}x_{2a}(t)=\bar{x}_{2a}$, we have, for $\eps>0$, that there exists $T>0$ such that $\forall \: t\ge T$ we have
\[\dot x_1(t)> x_1(t)(\eps -\kappa x_1(t)),\]
And using that $\lim_{t\to\infty}x_1(t)=0$, we have $\dot x_1(t)>0$ for $t$ large enough. This is a contradiction, and hence $(x_1(t),x_{2a}(t))$ cannot converge to $(0,\bar{x}_{2a})$ when $\kappa>\tilde{\kappa}$. 

\bigskip

It remains to check the cases when the $\omega$-limit set of our solution is a connected set consisting of a finite number of equilibria together with homoclinic and heteroclinic orbits connecting these equilibria.
In such a case, we have a sequence of non-equilibrium points $(q_n)_{n\in \N}$ included in such an orbit such that $\lim_{n\in\N}q_n=y$, where $y$ is an equilibrium. But for all $n$, we have a sequence of times $(t_m^n)_{m\in\N}$ such that the trajectory along those times converges to $q_n$ as $m$ tends to infinity. By extraction, we have a sequence of times such that the trajectory tends to $y$. As the solution gets closer and closer to $y$, we have that if $y$ is a stable equilibrium then the solution must stay close to it, and so we cannot have a homoclinic or heteroclinic orbit involving $y$. But if $y$ is $(1,0)$ (resp. $(0,\bar{x}_{2a})$) and $y$ is unstable, then as we have seen $\dot{x}_1$ is negative (resp. positive) in a neighborhood of $y$ in $\Delta$, and so the solution cannot get as close as we want to $y$, and this contradicts the existence of any heteroclinic or homoclinic orbit involving such a $y$. This excludes all the homoclinic or heteroclinic orbits involving a stable equilibrium or the equilibria $(1,0)$ and $(0,\bar{x}_{2a})$ when those ones are unstable. 

As we do not consider the case when both of these equilibra are stable, it remains to see what we get when both are unstable. Indeed, we do not have any result on the stability of the coordinate-wise positive equilibrium. But thanks to what we have proven, we just have two possibilities as time tends to $+\infty$: convergence to this equilibrium or a union of finitely many homoclinic orbits around this equilibrium. In case of a union of homoclinic orbits, we can extract a sequence of times such that the trajectory along this sequence converges to this equilibrium. 

According to the Poincaré--Bendixson theorem, since the system is two-dimensional and autonomous and the bounded set $\overline{\Delta}$ is positively invariant under it, any trajectory starting from $\overline{\Delta}$ must converge to a connected set consisting of a finite number of equilibria together with homoclinic and heteroclinic orbits between them, a periodic orbit (limit cycle), or an equilibrium. Lemma~\ref{lemma-convtoaneq} excludes the existence of a periodic orbit.
Using the proof of this lemma, we exclude all the cases that are not contained in the statements of the lemma, and we conclude the proof.  
\end{proof}

To end this section, we provide a conjecture on the behaviour of the dynamical system~\eqref{2dMoran} in the cases that we were not able to treat fully in the present section.

\begin{conj}\label{conj-Moran}
\begin{enumerate}[(A)]
    \item For $\sigma > \frac{2s^2}{p(1-s))}$, if $\kappa\in(\tilde{\kappa},\kappa') $, the equilibrium $(\tilde x_1,\tilde x_{2a})$ of the system~\eqref{2dMoran} is locally asymptotically stable. Consequently, started from a coordinate-wise positive initial condition $(x_1(0),x_{2a}(0))$ such that $x_1(0)+x_{2a}(0)<1$, the system converges to $(\widetilde x_1,\widetilde x_{2a})$. 
    \item For $\sigma <\frac{2s^2}{p(1-s))}$, if $\kappa\in(\kappa',\tilde{\kappa})$, the equilibrium $(\tilde x_1,\tilde x_{2a})$ is unstable and both eigenvalues of the corresponding Jacobi matrix have a strictly positive real part.
\end{enumerate}
\end{conj}
Note that if the first \igb assertion in \bl (A) is true, then the second one follows immediately. Indeed, if $(\tilde x_1,\tilde x_{2a})$ is asymptotically stable, there can be no homoclinic orbit around it, and we have seen in the proof of Lemma~\ref{lemma-convtoaneq2} that  apart from the equilibrium $(\tilde x_1,\tilde x_{2a})$, only the union of such homoclinic orbits around this equilibrium can be a possible candidate for the $\omega$-limit set of any relevant solution. 

\section{Proof of Theorem~\ref{theorem-mainMoran}}\label{sec-proofMoran}
The goal of this section is to prove our main result, Theorem~\ref{theorem-mainMoran}. Sections~\ref{sec-phase1proofMoran}, \ref{sec-phase2proofMoran}, and~\ref{sec-phase3proofMoran} provide preliminary results on the first, second, and third phase of the invasion, respectively. We will use these, and the results of Section~\ref{sec-preliminaryproofsMoran}, altogether in order to complete the proof of the theorem in Section~\ref{sec-endofproofMoran}. 

\subsection{First phase of the invasion: extinction or growth of the mutant population}\label{sec-phase1proofMoran}
\begin{proof}

  When $X_1^{(N)}$ is close to the total population size \igb $N$, \bl  we \igb may \bl approximate the process $(X_{2a}^{(N)}(t),X_{2d}^{(N)}(t))$ by the branching process $(Z_{2a}^{(N)}(t),Z_{2d}^{(N)}(t))$ introduced in Section~\ref{sec-3phasesMoran}.
Let us recall that $q_{2a}$ denotes the extinction probability of this branching process. Recall also from~\eqref{Teps2defMoran} that
\[T_{\eps}^2=\inf\{t\ge0: X_2^{(N)}(t)=\lfloor\eps N\rfloor\} \] 
and similarly
\[ T_{\eps}^{1} = \inf \{ t \geq 0 \colon X_1^{(N)}(t)=\lfloor \eps N \rfloor \}. \]
We verify the following lemma about the first phase of the invasion. Its statement is essentially analogous to the one of~\cite[Lemma 3.2]{C+19}, and its proof is also \igb similar; nevertheless we provide it here in \bl full detail for completeness \igb and accessibility. Note also the \bl  mathematically but also biologically \igb related result in \bl \cite[Proposition 4.1]{BT19}. \bl
\begin{lemma}\label{lemma-phase1Moran}
Under the assumptions of Theorem~\ref{theorem-mainMoran}, there exists a function $f \colon (0,\infty) \to (0,\infty)$ tending to zero as $\eps \downarrow 0$ such that
\[\begin{aligned}
    \liminf_{N\to\infty}\mathbb{P}\Big( \frac{T_0^2}{\log(N)}\le f(\eps); T_0^2<T_{\eps}^2 \Big) = q_{2a}-o_{\eps}(1) 
\end{aligned} \numberthis\label{extinctionphase1Moran} \]
and
    \[\limsup_{N\to\infty}\Big|\mathbb{P}\Big(T_{\eps}^2\le T_0^2; \: \big|\frac{T_{\eps}^2}{\log N}-\frac{1}{\lambda_2} \big|\le f(\eps)\Big)-(1-q_{2a})\Big|=o_{\eps}(1), \numberthis\label{goalinvasionMoran}\] 
    where $o_{\eps}(1)$ tends to $0$ as $\eps \downarrow 0$. 

\end{lemma}
\begin{proof}
We first show that 
\[ \limsup_{N\to\infty}\Big| \mathbb{P}\Big( T_0^2<T_{\eps}^2 \Big| \frac{1}{N}\mathbf X^{(N)}(0)=\big( \frac{N-1}{N},\frac{1}{N},0\big) \Big) -q_{2a}\Big| =o_{\eps}(1). \numberthis\label{phase1} \]
To do so, we introduce the branching process $(Z_{2a}^{\eps, -}(t),Z_{2d}^{\eps,-}(t))$ that jumps from $(i,j) \in \mathbb N_0^2$ to:
\begin{itemize}
    \item $(i-1,j)$ at rate $i\frac{1+s}{2}(1-p)c_-$,
    \item $(i+1,j)$ at rate $i \frac{1-s}{2}(1-\eps)$,
    \item $(i-1,j+1)$ at rate $pi(\frac{1+s}{2}(1-\eps)-\eps)$,
    \item $(i+1,j-1)$ at rate $\sigma j$,
    \item $(i,j-1)$ at rate $\kappa j(1+\eps)$\color{black}, 
\end{itemize}
with $c_-=1+\frac{\eps p}{1-p}+\frac{2p\eps}{(1+s)(1-p)}$, and the branching process $(Z_{2a}^{\eps,+}(t),Z_{2d}^{\eps,+}(t))$ that jumps from $(i,j) \in \mathbb N_0^2$ to:
\begin{itemize}
    \item $(i-1,j)$ at rate $i\frac{1+s}{2}(1-p)c_+$,
    \item $(i+1,j)$ at rate $i \frac{1-s}{2}(1+\eps)$,
    \item $(i-1,j+1)$ at rate $pi(\frac{1+s}{2}(1+\eps)+\eps)$,
    \item $(i+1,j-1)$ at rate $\sigma j+\kappa j\eps$, 
\end{itemize}
with $c_+=1-\frac{p\eps}{1-p}$, to obtain the almost \igb sure simultaneous sandwich couplings \bl 
\[
\begin{aligned}
    Z_{2v}^{\eps,-}(t)\le Z_{2v}^{(N)}(t)&\le Z_{2v}^{\eps,+}(t), \\[.1cm]
    Z^{\eps,-}_{2v}(t)\le X^{(N)}_{2v}(t)&\le Z^{\eps,+}_{2v}(t), 
\end{aligned}
\numberthis\label{firstcouplingMoran}
\]
\igb on $[0,T_0^2 \wedge T_\eps^2]$ for $\eps$ small enough,
where $v\in\{a,d\}$. \bl Indeed, the processes appearing in~\eqref{firstcouplingMoran} can be constructed on the same probability space including a suitable Poissonian construction, which we will formally introduce below in the proof of Lemma~\ref{lemma-phase2Moran}. \color{black} 

We write $q^{(\eps,\diamond)}$ \igb for \bl the extinction probability of $(Z_{2a}^{\eps,\diamond}(t),Z_{2d}^{\eps,\diamond}(t))$ for $\diamond\in\{-,+\}$ and $\eps>0$ fixed. As in~\cite[p.~487]{C+19}, we obtain:
\[0\le\liminf_{\eps\to0}|q^{(\eps,\diamond)}-q_{2a}|\le\limsup_{\eps\to0}|q^{(\eps,\diamond)}-q_{2a}|\le\limsup_{\eps\to0}|q^{(\eps,-)}-q^{(\eps,+)}|.\]

We introduce also $T_x^{(\eps,\diamond),2}=\inf\{t>0:Z_{2a}^{(\eps,\diamond)}(t)+Z_{2d}^{(\eps,\diamond)}(t)=\lfloor Nx\rfloor\}$. Thanks to the \igb sandwich \bl coupling, we have: 
\[\mathbb{P}(T_{\eps}^{(\eps,-),2} \le T_{0}^{(\eps,-),2}) \le \mathbb{P}(T_{\eps}^{2} \le T_{0}^{2}) \le\mathbb{P}(T_{\eps}^{(\eps,+),2} \le T_{0}^{(\eps,+),2}),\] and moreover:
\[\begin{aligned}
    \liminf_{N\to\infty}\mathbb{P}(T_{\eps}^{(\eps,-),2} \le T_{0}^{(\eps,-),2})\ge 1-q^{(\eps,-)}, \\
    \limsup_{N\to\infty}\mathbb{P}(T_{\eps}^{(\eps,+),2} \le T_{0}^{(\eps,+),2})\le 1-q^{(\eps,+)},
\end{aligned}\]
and hence: 
\[ \limsup_{N\to\infty}|\mathbb{P}(T_{\eps}^2\le T_0^2)-(1-q_{2a})|=o_{\eps}(1), \] and we obtain~\eqref{phase1} by taking \igb complements. \bl

\bigskip

\igb Now, let $B$ the event that \bl 
the branching process $(Z_{2a}^{(N)}(t),Z_{2d}^{(N)}(t))$
\igb goes extinct. Moreover, let $f$ be a \bl
continuous function such that $f(\eps)\xrightarrow[\eps\to0]{}0$. Using that
\[\limsup_{N\to\infty}\mathbb{P}\big(B\triangle\{T_0^2<T_{\eps}^2\}\big) +\mathbb{P}\big(B\triangle\{T_0^{(\eps,+),2}<\infty\}\big)=o_{\eps}(1),\]
where $\triangle$ means symmetric difference, we can end the proof of the result on the first phase of the invasion:
\[\begin{aligned}
    &\ge \liminf_{N\to\infty}\mathbb{P}\Big(\frac{T_0^2}{\log(N)}\le f(\eps); T_0^{(\eps,+)2}<\infty\Big) -o_{\eps}(1) \\
    &\ge \liminf_{N\to\infty}\mathbb{P}\Big(\frac{T_0^{(\eps,+),2}}{\log(N)}\le f(\eps); T_0^{(\eps,+)2}<\infty\Big) -o_{\eps}(1) \\
    &\ge \liminf_{N\to\infty}\mathbb{P}(T_0^{(\eps,+)2}<\infty) -o_{\eps}(1) \\
    &\ge q_{2a}-o_{\eps}(1),
\end{aligned}\]
and thus we have verified~\eqref{extinctionphase1Moran}. 
\igb Hence, \bl on the event $\{ T_0^2 < T_{\eps}^2\}$, we have
\[ \lim_{N\to\infty} \frac{T_0^2}{\log N} = 0 \qquad \text{in probability}. \numberthis\label{sublogagain} \]

\bigskip

Let us recall, for the rest of the proof, that we have
\[\limsup_{N\to\infty}|\mathbb{P}(T_{\eps}^2\le T_0^2)-(1-q_{2a})|=o_{\eps}(1). \]
Moreover, we will see that 
$$
\limsup_{N\to\infty}\Big| \mathbb{P} \Big( T_{\eps}^2\le T_0^2; \: |\frac{T_{\eps}^2}{\log N }-\frac{1}{\lambda_2}|\le f(\eps) \Big)-(1-q_{2a}) \Big| = o_{\eps}(1),
$$ 
where $f$ is a function that satisfies the same conditions as before.
Indeed, let us denote \igb by \bl $\tilde{\lambda}^{(\eps, \diamond})$ \igb the maximal \bl eigenvalue of the mean matrix of $(Z_{2a}^{(\eps,\diamond)}, Z_{2d}^{(\eps,\diamond)})$. This eigenvalue is positive for $\eps$ small enough and tends to $\lambda_2$ as $\eps\to0$. As before:
\[\begin{aligned}
    \mathbb{P}\Big( T_{\eps}^{(\eps,-),2}\le T_0^{(\eps,-),2} \wedge \frac{\log N}{\lambda_2}(1+f(\eps)) \Big)&\le \mathbb{P} \Big( T_{\eps}^{2}\le T_0^{2} \wedge \frac{\log N}{\lambda_2}(1+f(\eps)) \Big) \\
    &\le \mathbb{P}\Big( T_{\eps}^{(\eps,+),2}\le T_0^{(\eps,+),2} \wedge \frac{\log N}{\lambda_2}(1+f(\eps))\Big), 
\end{aligned}\]
for $\eps$ small enough such that $|\frac{\tilde{\lambda}^{(\eps,\diamond)}}{\lambda_2}-1|\le \frac{f(\eps)}{2}$. 

For $\eps$ small enough, we can argue analogously to~\cite[equation (4.34)]{BT21} to obtain
\[\limsup_{N\to\infty}\Big|\mathbb{P}\Big( T_{\eps}^2\le T_0^2\wedge\frac{\log N}{\lambda_2}(1+f(\eps))\Big)-(1-q_{2a})\Big|=o_{\eps}(1),\] 
and therefore
\[\limsup_{N\to\infty}\Big|\mathbb{P}\Big(T_{\eps}^2\le T_0^2; \: |\frac{T_{\eps}^2}{\log N}-\frac{1}{\lambda_2}|\le f(\eps)\Big)-(1-q_{2a})\Big|=o_{\eps}(1),\] 
which is~\eqref{goalinvasionMoran}, as wanted.
\end{proof}


\subsection{Second phase of the invasion: getting close to the new equilibrium}\label{sec-phase2proofMoran}
Let us now consider the second phase, where the number of resident individuals starts to decrease and the population gets closer to its new equilibrium (which is either $(0,\bar x_{2a})$ or $(\tilde x_1,\tilde x_{2a})$, depending on the choice of parameters). Note that according to Lemma~\ref{lemma-phase1Moran}, in case $(1,0)$ is hyperbolically asymptotically stable, or equivalently, $q_{2a}=1$, i.e.\  $\kappa> \kappa'=\frac{\sigma}{s}(\frac{1+s}{2}p-s)$, type $2$ dies out in the first phase with high probability and therefore no second phase occurs.


\medskip

Using the Kesten--Stigum theorem \igb for multi-type branching processes \bl (see e.g.\ \cite[Theorem 2.1]{GB03}), when $\lambda_2>0$ (which is equivalent to $\kappa< \kappa'$, i.e.\ \eqref{2invades1Moran}), we have:
\[
\Big(\frac{Z_{2a}^{(N)}(t)}{Z_2^{(N)}(t)},\frac{Z_{2d}^{(N)}(t)}{Z_2^{(N)}(t)}\Big)\xrightarrow[t\to\infty]{}(\pi_{2a},\pi_{2d}),
\]
almost surely conditionally on the event of survival of the branching process, where $Z_2^{(N)}(t)=Z_{2a}^{(N)}(t)+Z_{2d}^{(N)}(t)$ and $(\pi_{2a},\pi_{2d})$ is the positive left eigenvector associated to $\lambda_2$ such that $\pi_{2a}+\pi_{2d}=1$.

We first need to prove the following lemma in order to use Lemma \ref{lemma-convtoaneq} below. 
\begin{lemma}\label{lemma-phase2Moran}
Assume that~\eqref{2invades1Moran} holds. Then for 
$C$ sufficiently large and for $\delta>0$ such that $\pi_{2a} \pm \delta \in (0,1)$, we have:
\[\liminf_{N\to\infty}\mathbb{P} \Big( \exists \: t \in [T_{\eps}^2,T_{\sqrt{\eps}}^2]; \: \frac{\eps N}{C}\le X_2^{(N)}(t)\le \sqrt{\eps}N; \: \pi_{2a}-\delta < \frac{X_{2a}^{(N)}(t)}{X_2^{(N)}(t)}< \pi_{2a}+\delta \Big| T_{\sqrt{\eps}}^2< T_0^2 \Big) \ge 1 -o_{\eps}(1).\]
\end{lemma}
To a great extent, this lemma is analogous to~\cite[Proposition 3.2]{C+19} and \cite[Proposition 4.4]{BT19}, and the principle of the proof is also rather similar. Still, some arguments, especially the ones related to the semimartingale decomposition of the mutant active/dormant proportions, have to be adapted to our setting. \igb Moreover, \bl the way we will use this lemma will be different from the aforementioned papers, \igb due \bl to the differences in the behaviour of the corresponding dynamical systems.
\begin{proof}[Proof of Lemma~\ref{lemma-phase2Moran}]
As in~\cite[Proposition 4.4]{BT19}, if $\pi_{2a}-\delta < \frac{X_{2a}^{(N)}(T_{\eps}^2)}{X_2^{(N)}(T_{\eps}^2)}< \pi_{2a}+\delta$, there is nothing to show. We then also suppose that $\frac{X_{2a}^{(N)}(T_{\eps}^2)}{X_2^{(N)}(T_{\eps}^2)}\le \pi_{2a}-\delta$ and remark that the case $\frac{X_{2a}^{(N)}(T_{\eps}^2)}{X_2^{(N)}(T_{\eps}^2)}\ge \pi_{2a}+\delta$ will be similar.

First of all, it is clear that Lemma~\ref{lemma-phase1Moran} also holds with $\eps$ replaced by $\sqrt\eps$ everywhere on the left-hand side, i.e.\ under the conditions of \igb this \bl lemma, we can also find \igb a suitable \bl $f$ 
such that
\[\limsup_{N\to\infty}\Big|\mathbb{P}\Big(T_{\sqrt\eps}^2\le T_0^2; \: |\frac{T_{\sqrt\eps}^2}{\log N}-\frac{1}{\lambda_2}|\le f(\eps)\Big)-(1-q_{2a})\Big|=o_{\eps}(1). \] 

The \igb expression for the $\liminf$ \bl in the lemma told us that after reaching $T_{\eps}^2$, with high probability, for $\eps$ small enough, the mutant population will not \igb shrink \bl  below $\frac{\eps N}{C}$ for $C$ sufficiently large. We introduce the stopping time
\[T_{\eps, \eps /C}=\inf\big\{t\ge T_{\eps}^2; \: X_2^{(N)}(t)\le \frac{\eps N}{C} \big\}.\]
\igb Using\bl~\cite[Lemma A.1]{C+19} and the supercriticality of $(Z_2^{(\eps, -)}(t))$ \igb yields \bl
\[\lim_{N\to\infty}\mathbb{P} \big( T_{\eps, \eps /C} < T_{\sqrt{\eps}}^2 | T_{\sqrt{\eps}}^2<T_0^2 \big)=0.\]
\igb Further, \bl by \cite[Lemma A.2]{C+19}, using the stochastic domination of the total size of mutant individuals by a Yule process 
\igb with \bl birth rate $\frac{1-s}{2}$ (by considering the second type of jump in the definition of the process $X^{(N)}$), we have
\[\lim_{N\to\infty}\mathbb{P} \big(T_{\sqrt{\eps}}^2\le T_{\eps}^2 +\log\log(1/\eps) \, \big| \, T_{\sqrt{\eps}}^2<T_0^2\big) \le \sqrt{\eps}(\log(1/\eps))^{\frac{1-s}{2}}.\]

To show that with probability close to one our fraction $\frac{X_{2a}^{(N)}(t)}{X_2^{(N)}(t)}$ cannot stay below $\pi_{2a}-\delta$ on $[T_{\eps}^2,T_{\sqrt{\eps}}^2]$, we introduce five independent Poisson random measures $P_i$ on $[0,\infty]^2$ with intensity $dud\theta$, corresponding to all the possible transitions for the process $X^{(N)}$. We denote by $\tilde{P}_i(du,d\theta)=P_i(du,d\theta)-dud\theta$ the associated compensated measures.

Then, $\frac{X_{2a}^{(N)}(t)}{X_2^{(N)}(t)}$ is a semimartingale that we can rewrite as
\[\frac{X_{2a}^{(N)}(t)}{X_2^{(N)}(t)}=\frac{X_{2a}^{(N)}(T_{\eps}^2)}{X_2^{(N)}(T_{\eps}^2)}+M_2(t)+V_2(t),\]
for $t\ge T_{\eps}^2$, $M_2$ a martingale and $V_2$ a finite variation process such that
\[\begin{aligned}
    M_2(t)=&-\int_{T_{\eps}^2}^t \int_{[0,\infty[} \mathds{1}_{\big\{\theta\le\frac{1}{N}X_1^{(N)}(u-)X_{2a}^{(N)}(u-)\frac{1+s}{2}(1-p)\big\}}\frac{X_{2d}^{(N)}(u-)}{X_2^{(N)}(u-)(X_2^{(N)}(u-)-1)}\tilde{P}_1(du,d\theta) \\
    &+\int_{T_{\eps}^2}^t \int_{[0,\infty[} \mathds{1}_{\big\{\theta\le\frac{1}{N}X_1^{(N)}(u-)X_{2a}^{(N)}(u-)\frac{1-s}{2}\big\}}\frac{X_{2d}^{(N)}(u-)}{X_2^{(N)}(u-)(X_2^{(N)}(u-)+1)}\tilde{P}_2(du,d\theta) \\
    &-\int_{T_{\eps}^2}^t \int_{[0,\infty[} \mathds{1}_{\big\{\theta\le\frac{1}{N}p(X_1^{(N)}(u-)X_{2a}^{(N)}(u-)\frac{1+s}{2}+X_{2a}^{(N)}(u-)(X_{2a}^{(N)}(u-)-1))\big\}}\frac{1}{X_2^{(N)}(u-)}\tilde{P}_3(du,d\theta) \\
    &+\int_{T_{\eps}^2}^t \int_{[0,\infty[} \mathds{1}_{\big\{\theta\le\sigma X_{2d}^{(N)}(u-)  + \kappa \frac{X_{2d}^{(N)}(u-)X_{2a}^{N}(u-)}{N} \big\}}\frac{1}{X_2^{(N)}(u-)}\tilde{P}_4(du,d\theta) \\
    &+ \int_{T_{\eps}^2}^t \int_{[0,\infty[} \mathds{1}_{\big\{\theta\le \kappa\frac{X_1^{(N)}(u-)X_{2d}^{(N)}(u-)}{N}\big\}}\frac{X_{2a}^{(N)}(u-)}{X_2^{(N)}(u-)(X_2^{(N)}(u-)-1)}\tilde{P}_5(du,d\theta), 
\end{aligned}\]
and 
\[\begin{aligned}
    V_2(t)=&-\int_{T_{\eps}^2}^t \frac{1}{N}X_1^{(N)}(u)X_{2a}^{(N)}(u)\frac{1+s}{2}(1-p)\frac{X_{2d}^{(N)}(u)}{X_2^{(N)}(u)(X_2^{(N)}(u)-1)}du \\
    &+\int_{T_{\eps}^2}^t \frac{1}{N}X_1^{(N)}(u)X_{2a}^{(N)}(u)\frac{1-s}{2}\frac{X_{2d}^{(N)}(u)}{X_2^{(N)}(u)(X_2^{(N)}(u)+1)}du \\
    &-\int_{T_{\eps}^2}^t \frac{1}{N}p(X_1^{(N)}(u)X_{2a}^{(N)}(u)\frac{1+s}{2}+X_{2a}^{(N)}(u)(X_{2a}^{(N)}(u)-1))\frac{1}{X_2^{(N)}(u)}du \\
    &+\int_{T_{\eps}^2}^t \Big(\sigma X_{2d}^{(N)}(u)+ \kappa\frac{ X_{2d}^{(N)}(u)X_{2a}^{(N)}(u)}{N}\Big) \frac{1}{X_2^{(N)}(u)}du \\
    &+\int_{T_{\eps}^2}^t \kappa\frac{X_{1}^{(N)}(u)X_{2d}^{(N)}(u)}{N}\frac{X_{2a}^{(N)}(u)}{X_2^{(N)}(u)(X_2^{(N)}(u)-1)}du.
\end{aligned}\]

For $M_2$, we have the following predictable quadratic variation:
\[\begin{aligned}
    \langle M_2\rangle_t=&-\int_{T_{\eps}^2}^t \frac{1}{N}X_1^{(N)}(u)X_{2a}^{(N)}(u)\frac{1+s}{2}(1-p)\frac{(X_{2d}^{(N)}(u))^2}{(X_2^{(N)}(u))^2(X_2^{(N)}(u)-1)^2}du \\
    &+\int_{T_{\eps}^2}^t \frac{1}{N}X_1^{(N)}(u)X_{2a}^{(N)}(u)\frac{1-s}{2}\frac{(X_{2d}^{(N)}(u))^2}{(X_2^{(N)}(u))^2(X_2^{(N)}(u)+1)^2}du \\
    &-\int_{T_{\eps}^2}^t \frac{1}{N}p\Big(X_1^{(N)}(u)X_{2a}^{(N)}(u)\frac{1+s}{2}+X_{2a}^{(N)}(u)(X_{2a}^{(N)}(u)-1)\Big)\frac{1}{(X_2^{(N)}(u))^2}du \\
    &+\int_{T_{\eps}^2}^t \Big(\sigma X_{2d}^{(N)}(u)+\kappa\frac{X_{2d}^{(N)}(u)X_{2a}^{(N)}(u)}{N}\Big)\frac{1}{(X_2^{(N)}(u))^2}du \\
    &+ \int_{T_{\eps}^2}^t \kappa\frac{X_{1}^{(N)}(u)X_{2d}^{(N)}(u)}{N}\frac{(X_{2a}^{(N)}(u))^2}{(X_2^{(N)}(u)(X_2^{(N)}(u)-1))^2}du.
\end{aligned}\]
And so, $\exists \: C_0>0$ such that $\forall \: t \ge T_{\eps}^2$: $\langle M_2 \rangle_t\le C_0(t-T_{\eps}^2)\sup_{T_{\eps}^2\le u \le t}\frac{1}{X_2^{(N)}(u)-1}$. This leads us to
\[\langle M_2\rangle_{(T_{\eps}^2+\log\log(1/\eps))\wedge T_{\eps, \eps/C}}\le \frac{C_0\log\log(1/\eps)}{\frac{\eps N}{C}-1},\]
and using this, as in~\cite{C+19}, 
\[
\begin{aligned}
   & \limsup_{N\to\infty}\mathbb{P}\Big( \sup_{T_{\eps}^2\le t\le T_{\eps}^2+\log\log(1/\eps)}|M_2(t)|\ge \eps \: |\: T_{\sqrt{\eps}}^2<T_0^2\Big) \\ &\le \limsup_{N\to\infty}\Big[ \mathbb{P}\Big(\sup_{T_{\eps}^2\le t\le (T_{\eps}^2+\log\log(1/\eps))\wedge T_{\eps, \eps/C}}|M_2(t)|\ge \eps \: |\: T_{\sqrt{\eps}}^2<T_0^2\Big) \\& \quad +\mathbb{P} \Big(T_{\eps, \eps/C} < T_{\eps}^2+\log\log(1/\eps) \: | \:  T_{\sqrt{\eps}}^2<T_0^2\Big)\Big] \\
    &\le \limsup_{N\to\infty} \frac{1}{\eps^2}\mathbb{E}\Big[\langle M_2 \rangle_{(T_{\eps}^2+\log\log(1/\eps))\wedge T_{\eps, \eps/C}} \: \ \: T_{\sqrt{\eps}}^2<T_0^2\Big] + \sqrt{\eps}\log(\frac{1}{\eps})^{\frac{1-s}{2}} \\
    &= \sqrt{\eps}\log\big(1/\eps\big)^{\frac{1-s}{2}}.
\end{aligned}
\]

For the finite variation process, we rewrite it as
\[V_2(t)=\int_{T_{\eps}^2}^t P^{(u)}\Big(\frac{X_{2a}^{(N)}(u)}{X_2^{(N)}(u)}\Big)\frac{X_2^{(N)}(u)}{X_2^{(N)}(u)+1} + Q^{(u)}\Big(\frac{X_{2a}^{(N)}(u)}{X_2^{(N)}(u)}\Big)\frac{X_2^{(N)}(u)}{X_2^{(N)}(u)-1} + R^{(u)}\Big(\frac{X_{2a}^{(N)}(u)}{X_2^{(N)}(u)}\Big) \: du,\]
where 
\[\begin{aligned} P^{(u)}(x)& =\frac{1}{N}X_1^{(N)}(u)\frac{1-s}{2}(1-x)x, \\ Q^{(u)}(x)&=\Big( \kappa\frac{X_1^{(N)}(u)}{N}-\frac{1}{N}X_1^{(N)}(u)\frac{1+s}{2}(1-p)\Big)x(1-x), \\ R^{(u)}(x)&=-\frac{1}{N}p\Big(X_1^{(N)}(u)x\frac{1+s}{2} +x(X_{2a}^{(N)}(u)-1)\Big)+\Big(\sigma +\kappa\frac{X_{2a}^{(N)}(u)}{N}\Big)(1-x). \end{aligned}\]
Then, for $N$ large enough and $\eps$ small enough, on $[T_{\eps}^2, T_{\sqrt{\eps}}^2]$, those functions are close on $[0,1]$ to the polynomial functions

\[P(x)=\frac{1-s}{2}(1-x)x, \quad Q(x)=\big(\kappa-\frac{1+s}{2}(1-p)\big)(1-x)x, \quad R(x)= \sigma(1-x)-px\frac{1+s}{2}.\]
And so, for $\eps$ small enough, the integrand of $V_2$ is close to the function

\[S(x)=x(1-x)\Big( \frac{1-s}{2} - \frac{1+s}{2}(1-p)+\kappa \Big) + \sigma(1-x)- px\frac{1+s}{2}.\]
And we can rewrite $S$ as: 
\[S(x)=-x^2\Big( p\frac{1+s}{2}-s+\kappa \Big) - x(s+\sigma -\kappa )+\sigma.\]
Then we have $S(0)>0$, $S(1)<0$, and as the degree of $S$ is two, we know that $x'=S(x)$ has a unique equilibrium \igb in \bl  $[0,1]$. Moreover, this equilibrium is $\pi_{2a}$. Indeed, we have that $(\pi_{2a},\pi_{2d})$ is a left eigenvector associated to $\lambda_2$ for the matrix $J$, and that $\pi_{2d}=1-\pi_{2a}$. Then,
\[-s\pi_{2a}+\sigma(1-\pi_{2a})=\pi_{2a}\frac{-(s+\sigma+\kappa)+\sqrt{(s+\sigma+\kappa)^2-4(s(\sigma+\kappa)-\frac{1+s}{2}p\sigma)}}{2}\]
and thus 
\[ \pi_{2a}=\frac{2\sigma}{\sqrt{(s+\sigma+\kappa)^2-4(s(\sigma+\kappa)-\frac{1+s}{2}p\sigma)}+\sigma+s-\kappa}\color{black}. \]
Using this, we can easily check that $\pi_{2a}$ is a root of $S$. 

We can then choose $\delta>0$ and $\theta>0$ such that $\pi_{2a}-\delta>0$ and $\forall\: x<\pi_{2a} - \delta$: $S(x)>\theta/2$. And so, for $\eps$ small enough and $N$ sufficiently large, we have for all $u \in [T_{\eps}^2,T_{\sqrt{\eps}}^2]$ and $x\in\: ]0,\pi_{2a}-\delta[$:
\[P^{(u)}(x)\frac{X_2^{(N)}}{X_2^{(N)}+1}+Q^{(u)}(x)\frac{X_2^{(N)}}{X_2^{(N)}-1}+R^{(u)}(x)\ge \frac{\theta}{2}>0.\]

We now define $\textbf{t}_{2a}^{(\eps)}=\inf\{t\ge T_{\eps}^2: \: \frac{X_{2a}^{(N)}(t)}{X_2^{(N)}(t)} \ge \pi_{2a}-\delta\}$. Using all the results we got before, we have that on $\{T_{\sqrt{\eps}}^2<T_0^2\}$ and $\forall \: t \in [T_{\eps}^2, (T_{\eps}^2+\log\log(1/\eps))\wedge \textbf{t}_{2a}^{(\eps)}]$ :

\[\pi_{2a}-\delta \ge \frac{X_{2a}^{(N)}(t)}{X_2^{(N)}(t)} \ge \frac{\theta}{2}\big(\log\log(1/\eps)\wedge(\textbf{t}_{2a}^{\eps}-T_{\eps}^2)\big)-\eps,\]
with a probability greater than $1-\sqrt{\eps}\log(1/\eps)^{\frac{1-s}{2}}$. And, as $\log\log(1/\eps)\xrightarrow[\eps \to 0]{}\infty$, we have that, for $\eps$ small enough, $\textbf{t}_{2a}^{\eps} \le T_{\eps}^2+\log\log(1/\eps)$ and so $\textbf{t}_{2a}^{\eps}\le T_{\sqrt{\eps}}^2$ with a probability close to one on $\{T_{\sqrt{\eps}}^2<T_0^2\}$.

Moreover, using that each jump of the fraction is less than $\frac{1}{\frac{\eps N}{C}+1}$, we get that for $N$ sufficiently large, the fraction is contained in the interval $[\pi_{2a}-\delta, \pi_{2a}+\delta]$. This finishes the proof of the lemma.
\end{proof}

\bigskip

According to Lemma~\ref{lemma-phase2Moran}, knowing that the mutants reach the level $\sqrt{\eps}N$ before dying out, with high probability \color{black} the population of mutants will not go below $\frac{\eps N}{C}$ for a $C$ sufficiently large after having reached the level $\eps N$. Moreover, the proportion of active mutants in the mutant population will be close to the equilibrium $(\pi_{2a},\pi_{2d})$ before the size of the mutant population reaches $\sqrt{\eps}N$. 

We will see in Section ~\ref{sec-endofproofMoran} that \color{black} the previous lemmas regarding the convergence of the dynamical system will, in case of invasion, bring us from the configuration state obtained in \color{black} Lemma~\ref{lemma-phase2Moran} to an initial condition corresponding to Lemma~\ref{lemma-phase3Moran}. 
In case of coexistence, Lemma~\ref{lemma-phase2Moran} and the lemmas on the dynamical system will in turn be useful for proving that the rescaled stochastic process approaches the polymorphic equilibrium with high probability conditional on mutant survival. 

\subsection{Third phase of the invasion: Extinction of the \igb formerly \bl resident population}\label{sec-phase3proofMoran}
Let us now study a configuration where the \igb mutant population size has grown close to its \bl equilibrium and the \igb originally resident population has shrunk to a size that \bl is small compared to the mutant one. We will estimate the time \igb to \bl extinction of the \igb latter \bl and show that \igb in the meantime, \bl the mutant population stays close to its equilibrium with high probability. We suppose here that $\bar{x}_{2a}>\frac{\kappa}{\kappa+\frac{1+s}{2}p-s}$ in order \igb the \bl the branching process associated to the \igb previously \bl resident population is subcritical (see ~\ref{phase3cond}).  Let us also recall that this condition is equivalent to the one for the stability of $(0,\bar{x}_{2a})$ which is $\kappa<\tilde\kappa$. 

\igb To carry out this program formally, we first define \bl 
$T_{S_{\eps}}=\inf\{t>0, \: X^{(N)}(t)\in S_{\eps}\}$, where
\[S_{\eps}=\{0\}\times[\bar{x}_{2a}-\eps, \bar{x}_{2a}+\eps]\times[\bar{x}_{2d}-\eps,\bar{x}_{2d}+\eps].\]

As in \cite[Proposition 4.9]{BT19}, we will prove the following assertion.
\begin{lemma}\label{lemma-phase3Moran}
Under the assumptions of Theorem~\ref{theorem-mainMoran} in case $\bar{x}_{2a}>\frac{\kappa}{\kappa+\frac{1+s}{2}p-s}$, there exist positive numbers $\eps_0$ and $C_0$ such that for \igb all \bl $\eps \in (0,\eps_0)$, if the initial condition satisfies
\[\big|\frac{X_{2a}^{(N)}(0)}{N}-\bar{x}_{2a}\big|\le\eps, \quad \big|\frac{X_{2d}^{(N)}(0)}{N}-\bar{x}_{2d} \big|\le\eps \quad \text{ and } \quad N\eta\eps \le X_1^{(N)}(0)\le \eps N,\]
for $\eta \in (0,\frac{1}{2})$, then we have
\[\begin{aligned}\forall \: \tilde{C}>\frac{1}{\lambda_1}+C_0\eps, \quad \lim_{N\to\infty} \mathbb{P}(T_{S_{\eps}}\le \tilde{C}\log(N))=1,  \\
\forall \: 0\le \tilde{C}<\frac{1}{\lambda_1}-C_0\eps, \quad \lim_{N\to\infty} \mathbb{P}(T_{S_{\eps}}\le \tilde{C}\log(N))=0.\end{aligned}\]
\end{lemma}
\begin{proof}
\igb We first again \bl use arguments from \cite[Section 4.1]{C+16} to show that the rescaled mutant population size $\frac1N (X^{(N)}_{2a}(t),X^{(N)}_{2d}(t))$ will stay close to its equilibrium $(\bar x_{2a},\bar x_{2d})$. 
\igb For $\eps>0$ and $i\in \{a,d\}$ consider the \bl
stopping times 
\[\begin{aligned}
    T_{\sqrt{\eps}}^1=\inf\{t\ge 0 \: : X_1^{(N)}(t)\ge \sqrt{\eps} N\}, \\
    T_0^1=\inf\{t\ge 0 \: : X_1^{(N)}(t)= 0\}, \\
    R_{\eps,i}=\inf\Big\{t\ge 0 \: : \Big| \frac{X_{2i}^{(N)}(t)}{N}-\bar{x}_{2i} \Big| > \eps \Big\}.
\end{aligned}\]
We introduce the processes $(Y_{2a}^{\eps,\le}(t),Y_{2d}^{\eps,\le}(t))$ and $(Y_{2a}^{\eps,\ge}(t),Y_{2d}^{\eps,\ge}(t))$ to obtain the \igb sandwich \bl coupling
\[ NY_{2i}^{\eps,\le}(t) \le X_{2i}^{(N)}(t) \le NY_{2i}^{\eps,\ge}(t) \quad \text{ a.s. for } i\in\{a,d\} \text{ and } 0\le t \le T_{\sqrt\eps}^1,\]
where the process $(Y_{2a}^{\eps, \le}(t),Y_{2d}^{\eps,\le}(t))$ jumps for $(i,j) \in \N_0^2$ from $(\frac{i}{N},\frac{j}{N})$ to:
\begin{itemize}
    \item $(\frac{i}{N}-\frac{1}{N},\frac{j}{N})$ at rate $\sqrt{\eps} i\frac{1+s}{2}$,
    \item $(\frac{i}{N}-\frac{1}{N},\frac{j}{N}+\frac{1}{N})$ at rate $\frac{pi(i-1)}{N}$,
    \item $(\frac{i}{N}+\frac{1}{N},\frac{j}{N}-\frac{1}{N})$ at rate $\sigma j$,
    \item $(\frac{i}{N},\frac{j}{N}-\frac{1}{N})$ at rate $ \kappa j (1+\sqrt{\eps}),$
\end{itemize}
and the process $(Y_{2a}^{\eps, \ge}(t),Y_{2d}^{\eps,\ge}(t))$ jumps from $(\frac{i}{N},\frac{j}{N})$ to:
\begin{itemize}
    \item $(\frac{i}{N}-\frac{1}{N},\frac{j}{N})$ at rate $i\frac{1+s}{2}$,
    \item $(\frac{i}N+\frac{1}{N},\frac{j}{N})$ at rate $\sqrt{\eps} i \frac{1-s}{2}$
    \item $(\frac{i}{N}-\frac{1}{N},\frac{j}{N}+\frac{1}{N})$ at rate $\frac{pi(i-1)}{N} + \sqrt{\eps} i \frac{1+s}{2}p$,
    \item $(\frac{i}{N}+\frac{1}{N},\frac{j}{N}-\frac{1}{N})$ at rate $\sigma j + \kappa j \sqrt{\eps}$. 
\end{itemize}

We introduce for these processes the stopping times
\[R_{\eta, v}^{\diamond}=\inf\{t\ge 0: Y_{2v}^{\eps,\diamond}(t) \notin [\bar{x}_{2v}-\eta, \bar{x}_{2v}+\eta]\},\]
where $\eta>0$, $v\in\{a,d\}$ and $\diamond \in \{\le,\ge\}$. For the process $(Y_{2a}^{\eps, \le}(t),Y_{2d}^{\eps,\le}(t))$, its dynamics is close, for $N$ large enough and $\eps$ small enough, to the one of the unique solution of the system
\[\begin{aligned}
    \dot y_{2a}(t)&=\sigma y_{2d}(t)-y_{2a}(t)\big(p(y_{2a}(t)-\sqrt{\eps}) +\sqrt{\eps}\frac{1+s}{2}\big), \\
    \dot y_{2d}(t)&=py_{2a}(t)(y_{2a}(t)-\sqrt{\eps}) - (\sigma +\kappa(1+\sqrt{\eps}))y_{2d}(t). \quad \qquad 
\end{aligned}\]
So, for all $\eps$ small enough, this system has a unique positive equilibrium (denoted by $(\bar{y}_{2a}^{\eps, \le},\bar{y}_{2d}^{\eps, \le})$) that tends to $(\bar{x}_{2a},\bar{x}_{2d})$ when $\eps$ tends to $0$. Thanks to Lemma \ref{lemma-convtoaneq2}, we have, for $\eps$ small enough, the asymptotic stability \igb of, and thus convergence 
to, \bl 
this equilibrium 
for any nonnegative initial condition (except $(0,0)$). Moreover, we have $\eps_0'$ and $c_0$ such that, for all $\eps \in (0,\eps_0')$:
\[\forall \: i \in \{a,d\} : |\bar{y}_{2i}^{\eps,\le}-\bar{x}_{2i}|\le (c_0-1)\eps \quad \text{and } 0\notin[\bar{x}_{2i}-c_0\eps, \bar{x}_{2i}+c_0\eps].\]
\igb From this \bl 
we obtain, thanks to \cite[Chapter 5]{FW84}, that there exists \igb a \bl $V>0$ such that
\[\lim_{N\to\infty}\mathbb{P}\big( R_{c_0\eps,a}^{\le} > \exp(NV); R_{c_0\eps,d}^{\le} > \exp(NV)\big) =1.\]
Similarly, $\lim_{N\to\infty}\mathbb{P}(R_{c_0\eps,a}^{\ge} > \exp(NV); R_{c_0\eps,d}^{\ge} > \exp(NV))=1.$

\bigskip

As the total population \igb size \bl is constant, we have for our $c_0$ that for $\eps$ small enough $R_{c_0\eps,a}\wedge R_{c_0\eps,d} \le T_{\sqrt\eps}^1$, and as in \cite[Section 4.3]{BT19}, we note that on $\{ R_{c_0\eps,i} \le T_{\sqrt\eps}^1\}$, we have $R_{c_0\eps,i}\ge R_{c_0\eps,i}^{\le} \wedge R_{c_0\eps,i}^{\ge}$. And thus
\[\lim_{N\to\infty}\mathbb{P}\big( (R_{c_0\eps,a}\wedge R_{c_0\eps,d} \le \exp(NV)\wedge T_{\sqrt\eps}^1\big)=0. \numberthis\label{BD4}\]
This result shows that the mutant population stays close to its equilibrium on $[0,\exp(NV)\wedge T_{\sqrt\eps}^1]$ as $N$ tends to infinity. It remains for us to see that the resident population will \igb indeed \bl go extinct, and that the extinction takes place approximately after $\frac1{\lambda_1}\log N$ time \igb units. \bl  

\medskip

\igb To show this, we \bl introduce the processes $(Z_1^{\eps,\le}(t))_{t\ge0}$ and $(Z_1^{\eps,\ge}(t))_{t\ge0}$ with the following transitions from $n \in \mathbb{N}_0$ to
\begin{itemize}
    \item $n+1$ at rate $n((\bar{x}_{2a}-c_0\eps)\frac{1+s}{2}(1-p)+\kappa(\bar{x}_{2d}-c_0\eps))$,
    \item $n-1$ at rate $n(\bar{x}_{2a}+c_0\eps)\frac{1-s}{2}$, 
\end{itemize}
for $Z_1^{\eps,\le}$, and: 
\begin{itemize}
    \item $n+1$ at rate $n((\bar{x}_{2a}+c_0\eps)\frac{1+s}{2}(1-p)+\kappa(\bar{x}_{2d}+c_0\eps))$,
    \item $n-1$ at rate $n(\bar{x}_{2a}-c_0\eps)\frac{1-s}{2}$, 
\end{itemize}
for $Z_1^{\eps,\ge}$, where we suppose that we start $Z_1^{\eps,\le}$ in $\lfloor N\eta \eps \rfloor$ and $Z_1^{\eps,\ge}$ in $\lceil N \eps \rceil$, to obtain \igb yet another sandwich \bl coupling
\[Z_1^{\eps,\le}(t)\le X_1^{(N)}(t)\le Z_1^{\eps,\ge}(t) \quad \text{a.s.\  for } t \in I_{\eps}^N=\big[ 0,R_{c_0\eps,a}\wedge R_{c_0\eps,d}\wedge T_{\sqrt\eps}^1\big].\]

For $\eps$ small enough, these two processes are subcritical with growth rate $\bar{x}_{2a}(s-p\frac{1+s}{2})+\kappa \bar{x}_{2d} \pm \mathcal{O}(\eps)$. Using 
arguments analogous to \cite[Section 4.1]{C+16}, 
because the birth rate is less than the death rate (we are in the case where $p>\frac{2s}{1+s}$ and $\bar{x}_{2a}>\frac{\kappa}{\kappa+\frac{1+s}{2}p-s}$), we have
\[\begin{aligned}
    \forall \: \tilde{C}<\Big( \bar{x}_{2a}\big(p\frac{1+s}{2}-s\big)-\kappa\bar{x}_{2d}\Big) ^{-1} \quad \lim_{N\to\infty}\mathbb{P}\big(\mathcal{S}_0^{Z_1^{\eps,\diamond}} \le \tilde{C}\log(N)\big)=0,\\
    \forall \: \tilde{C}>\Big( \bar{x}_{2a}\big(p\frac{1+s}{2}-s\big) -\kappa\bar{x}_{2d}\Big)^{-1} \quad \lim_{N\to\infty}\mathbb{P}\big(\mathcal{S}_0^{Z_1^{\eps,\diamond}} \le \tilde{C}\log(N)\big)=1,
\end{aligned}\]
for $\diamond \in\{\le,\ge\}$, and $S_0^{\mathcal{N}}=\inf\{t\ge0: \mathcal{N}(t)=0\}$.

Moreover, for $Z_1^{\eps,\le}(0)=\lfloor \eta\eps N\rfloor$, we have for $\eps$ small enough that
\[ \lim_{N\to\infty}\mathbb{P}\big(\mathcal{S}_0^{Z_1^{\eps,\le}}>N\wedge S_{\lfloor \eps N\rfloor}^{Z_1^{\eps,\le}}\big)=0, \numberthis\label{BD3}\] 
where $S_{\lfloor \eps N\rfloor}^{Z_1^{\eps,\le}}$ is the first time when our process $Z_1^{\eps,\le}$ reaches $\lfloor \eps N\rfloor$. Using all of this, we have, similarly to \cite{BT19}, for $C\ge0$:
\[\begin{aligned}
    \mathbb{P}(T_0^1<C\log(N))-\mathbb{P}(\mathcal{S}_0^{Z_1^{\eps,\le}}< C\log(N)) &= \mathbb{P}(T_0^1<C\log(N)\le \mathcal{S}_0^{Z_1^{\eps,\le}}) \\
    &\le \mathbb{P}(T_0^1<C\log(N); \: R_{c_0\eps,a}\wedge R_{c_0\eps,d} < T_{\sqrt\eps}^1\wedge T_0^1) \\
    &+ \mathbb{P}(T_0^1<C\log(N); \: T_0^1\wedge R_{c_0\eps,a}\wedge R_{c_0\eps,d} > T_{\sqrt\eps}^1) \\
    &+ \mathbb{P}(T_0^1<C\log(N) \le S_0^{Z_1^{\eps,\le}}; \: T_0^1< R_{c_0\eps,a}\wedge R_{c_0\eps,d} \wedge T_{\sqrt\eps}^1) \\
    &\le \mathbb{P}(R_{c_0\eps,a}\wedge R_{c_0\eps,d} < T_{\sqrt\eps}^1 \wedge C\log(N)) + \mathbb{P}(T_0^1>T_{\sqrt\eps}^1) \\
    &\le \mathbb{P}(R_{c_0\eps,a}\wedge R_{c_0\eps,d} < T_{\sqrt\eps}^1 \wedge N) + \mathbb{P}(T_0^1>T_{\sqrt\eps}^1\wedge N) \\
    &\le \mathbb{P}(R_{c_0\eps,a}\wedge R_{c_0\eps,d} < T_{\sqrt\eps}^1 \wedge N) + \mathbb{P}(S_0^{Z_1^{\eps,\le}}>S_{\lfloor \eps N \rfloor}^{Z_1^{\eps,\le}}\wedge N), 
\end{aligned}\]
where we obtain the second inequality by \igb noting \bl that if $T_0^1 \in [0,I_{\eps}^N]$, then $T_0^1\ge S_0^{Z_1^{\eps,\le}}$. By ~\eqref{BD4} and ~\eqref{BD3}, we have
\[\limsup_{N\to\infty} \mathbb{P}(T_0^1<C\log(N)) \le \lim_{N\to\infty} \mathbb{P}(\mathcal{S}_0^{Z_1^{\eps,\le}} \le C\log(N)),\]
and using the same steps with $\mathcal{S}_0^{Z_1^{\eps,\ge}}$ \igb yields \bl 
\[\liminf_{N\to\infty} \mathbb{P}(T_0^1<C\log(N)) \ge \lim_{N\to\infty} \mathbb{P}(\mathcal{S}_0^{Z_1^{\eps,\ge}} \le C\log(N)),\]
which implies the lemma.
\end{proof}
This gives us that, in our setting, we reach extinction of the residents \igb within \bl  order $\log(N)$ \igb time units,  
\igb while \bl the mutant population stays close to its equilibrium. We now have \igb the necessary \bl results on all three phases of an invasion of mutant individuals against resident ones. 

\subsection{Completion of the proof of Theorem~\ref{theorem-mainMoran}}\label{sec-endofproofMoran} Let us \igb now \bl  put together the \igb findings of the \bl previous three subsections and the \igb preparatory \bl results of Section~\ref{sec-preliminaryproofsMoran} on the dynamical system~\eqref{2dMoran} to obtain the theorem.
First of all, note that we have already proved~\eqref{sublogMoran} in~\eqref{sublogagain}, and that in parts (a), (b), (c) of the theorem, the assertion of the theorem in case (i) (fixation of type 1) and (iv) (founder control) readily follows from Section~\ref{sec-3phasesMoran}. Hence, it remains to consider parts (b) and (c) and show that under their respective assumptions, assertions (ii) and (iii) hold.

We first present the proof when assertion (ii) should hold. 

\color{black} For $\eps>0$ and $\beta>0$, we introduce the sets
\begin{itemize}
    \item $\mathcal{B}_{\eps}^1=[\pi_{2a}-\delta, \pi_{2a}+\delta]\times[\frac{\eps N}{C}, \sqrt{\eps}N],$
    \item $\mathcal{B}_{\beta}^2=[0,N\beta]\times[N(\bar{x}_{2a}-\beta),N(\bar{x}_{2a}+\beta)]\times[N(\bar{x}_{2d}-\beta),N(\bar{x}_{2d}+\beta)],$
\end{itemize}
and the stopping times
\begin{itemize}
    \item $T_{\eps}'=\inf\{t\ge0: (\frac{X_{2a}^{(N)}(t)}{X_2^{(N)}(t)},X_2^{(N)}(t)) \in \mathcal{B}_{\eps}^1\},$
    \item $T_{\beta}''=\inf\{t\ge0: (X_1^{(N)}(t),X_{2a}^{(N)}(t),X_{2d}^{(N)}(t)) \in \mathcal{B}_{\beta}^2\}.$
\end{itemize}

We consider the (pre-limiting) fixation event 
$$A(N,\eps)=\Big\{ \Big|\frac{T_{S_{\beta}}}{\log(N)}-\big( \frac{1}{\lambda_1}+\frac{1}{\lambda_2} \big) \Big|\le f(\eps); \: T_{\mathcal{S}_{\beta}}<T_0^2 \Big\}.$$ 


On $\{T_{\mathcal{S}_{\beta}}<T_0^2\}$ we have $T_{\sqrt{\eps}}^2<T_{\beta}''<T_{\mathcal{S}_{\beta}}$, and if $T_{\eps}'<\infty$ then $T_{\eps}'<T_{\beta}''$ and thus
\[\begin{aligned}
    \liminf_{N\to\infty}\P(A(N,\eps))&\ge \liminf_{N\to\infty}\Big[ \mathbb{P}\Big(|\frac{T_{\eps}'}{\log(N)}-\frac{1}{\lambda_2}| \le \frac{f(\eps)}{3} \:; \: T_{\eps}'<T_0^2; \: T_{\sqrt{\eps}}^2<T_0^2\Big) \\
    &\times \inf_{x=(x_1,x_{2a},x_{2d}) \colon (\frac{x_{2a}}{x_2},x_2)\in\mathcal{B}_{\eps}^1}\mathbb{P}\Big( |\frac{T_{\beta}''<T_{\eps}'}{\log(N)}| \le \frac{f(\eps)}{3}\: ; \: T_{\beta}''<T_0^2 | X^{(N)}(0)=x\Big) \\
    &\times \inf_{x \in \mathcal{B}_{\beta}^2}\mathbb{P}\Big(|\frac{T_{\mathcal{S}_{\beta}}-T_{\beta}''}{\log (N)}-\frac{1}{\lambda_1}|\le \frac{f(\eps)}{3} \: ; \: T_{\mathcal{S}_{\beta}}<T_0^2 | X^{(N)}(0)=x\Big) \Big].
\end{aligned}\]
We need to show that the right-hand side is close to $1-q_{2a}$ when $\eps$ is small. For the first term we have, using the same reasoning that in \cite[proof of (3.61)]{C+19}:
\[\liminf_{N\to\infty} \mathbb{P}\Big(\Big|\frac{T_{\eps}'}{\log(N)}-\frac{1}{\lambda_2}\Big| \le \frac{f(\eps)}{3} \:; \: T_{\eps}'<T_0^2; \: T_{\sqrt{\eps}}^2<T_0^2\Big) \ge 1-q_{2a} +o_{\eps}(1). \]

For the second term, we write for $\textbf{m}=(m_1,m_{2a},m_{2d})\in [0,\infty[^3$, the unique solution $\textbf{x}^{\textbf{m}}$ of the dynamical system ~\eqref{2dMoran} with initial condition $\textbf{m}$. 
Thanks to the continuity of the flow of this dynamical system with respect to the initial condition and thanks to the convergence provided by Lemma \ref{lemma-convtoaneq2} (we suppose $\kappa<\tilde\kappa\wedge \kappa'$ for assertion (ii)),  we have \color{black} the convergence of the dynamical system to $(0,\bar{x}_{2a})$. Precisely, \color{black} for any $\beta$ that we choose small enough, we can find \color{black} $\eps_0$ and $\delta_0$ such that for all $\eps \in (0,\eps_0)$ and all $\delta \in (0,\delta_0)$, there exists $t_{\beta,\eps,\delta}$ such that $\forall \: t> t_{\beta,\eps,\delta}$:
\[|\textbf{x}^{\text{x}^0}(t)-(0,\bar{x}_{2a},\bar{x}_{2d})|\le \beta,\]
for any initial condition $\textbf{x}^0\in \mathcal{B}_{\eps}^1$. Then, using  \cite[Theorem 2.1, p.~456]{EK}, we obtain for $\eps<\eps_0$: 
\[\lim_{N \to\infty}\mathbb{P}\Big(T_{\beta}''-T_{\eps}'<t_{\beta,\eps,\delta} \: \Big| \: \Big(\frac{X_{2a}^{(N)}(0)}{X_2^{(N)}(0)},X_2^{(N)}(0)\Big)\in\mathcal{B}_{\eps}^1\Big)=1-o_{\eps}(1).\]
Therefore, the second term is close to one when $N$ tends to infinity and $\eps$ is small enough. 

\medskip

For the third term, using Lemma ~\ref{lemma-phase3Moran}, there exists $\beta_0>0$ such that for all $\beta<\beta_0$ and $\eps$ small enough:
\[\lim_{N\to\infty}\mathbb{P}\Big(\Big|\frac{T_{\mathcal{S}_{\beta}}-T_{\beta}''}{\log(N)}-\frac{1}{\lambda_1}\Big| \le \frac{f(\eps)}{3} \: \Big| \: X^{(N)}(0)\in\mathcal{B}_{\beta}^2\Big)=1-o_{\eps}(1).\]
Using all of this, we obtain:
\[\liminf_{N\to\infty}\mathbb{P}\Big(\Big|\frac{T_{S_{\beta}}}{\log(N)}-\Big(\frac{1}{\lambda_1}+\frac{1}{\lambda_2}\Big)\Big|\le f(\eps); \: T_{\mathcal{S}_{\beta}}<T_0^2\Big)\ge 1-q_{2a}-o_{\eps}(1).\] 

As we reach the extinction of the residents \igb within \bl  $\log(N)$ \igb time units \bl when we start from $\mathcal{B}_{\beta}^2$, we then have on the event $\{ T_0^1 < T_0^2 \}$:
\[ \lim_{N\to\infty} \frac{T_0^1}{\log N} = \frac{1}{\lambda_2} + \frac{1}{\lambda_1} \qquad \text{in probability}, \]
and we also \igb obtain \bl  the result regarding the convergence of the mutant population in the theorem as \igb keeping in mind that 
the mutant population stays as close as we want to its equilibrium. \bl

Similarly, by modifying the assumptions, we obtain the result stated in the theorem regarding $T_\beta^{\rm co}$ in case (iii). Indeed, with the proper assumptions, Lemma~\ref{lemma-convtoaneq2} gave us the convergence to $(\tilde{x}_1,\tilde{x}_{2a})$ for the second phase.

\igb Note that the \bl  branching process associated to the residents we used for the third phase is no more subcritical here. What will happen in this situation was stated in Remark ~\ref{sec-metastability}.

\end{proof}

\subsection*{Acknowledgements.}  BLD was partially supported
by the ERC Synergy Grant No.\ 810115 -- DYNASNET\color{black}. Funding acknowledgement by BLD and AT: Project no.\ STARTING 149835 has been implemented with the support provided by the Ministry of Culture and Innovation of Hungary from the National Research, Development and Innovation Fund, financed under the STARTING\_24 funding scheme. \color{black} Further funding acknowledgement by AT: This paper was supported by the János Bolyai Research Scholarship of the Hungarian Academy of Sciences. \color{black} AT also wishes to thank the Goethe University Frankfurt for hospitality. 

\smallskip



\end{document}